\documentclass[a4paper,UKenglish,cleveref, autoref, thm-restate,notab, nolineno]{socg-lipics-v2021}
\hideLIPIcs 

\usepackage{cite}
\usepackage{booktabs}
\usepackage{nicematrix}

\newtheorem{problem}[theorem]{Problem}
\newtheorem*{problem*}{Problem}
\newcommand{\R}{\mathbb{R}}
\DeclareMathOperator{\ppolylog}{polylog}
\DeclareMathOperator{\ppoly}{poly}

\renewcommand{\epsilon}{\varepsilon}

\title{Conditional Lower Bounds for Dynamic Geometric Measure Problems}

\author{Justin Dallant}{Université libre de Bruxelles, Belgium} {Justin.Dallant@ulb.be}{https://orcid.org/0000-0001-5539-9037}{Supported by the French Community of Belgium via the funding of a FRIA grant.} 

\author{John Iacono}{Université libre de Bruxelles, Belgium} {john@johniacono.com}{https://orcid.org/0000-0001-8885-8172}{Supported by the Fonds de la Recherche Scientifique-FNRS under Grant no MISU F 6001 1.}

\authorrunning{J. Dallant and J. Iacono}

\Copyright{Justin Dallant and John Iacono}



\acknowledgements{We thank Jean Cardinal for helpful discussions about the topic of this paper.}

\begin{document}

\maketitle

\begin{abstract}
    We give new polynomial lower bounds for a number of dynamic measure problems in computational geometry. These lower bounds hold in the Word-RAM model, conditioned on the hardness of either 3SUM, APSP, or the Online Matrix-Vector Multiplication problem [Henzinger et al., STOC 2015]. In particular we get lower bounds in the incremental and fully-dynamic settings for counting maximal or extremal points in $\R^3$, different variants of Klee's Measure Problem, problems related to finding the largest empty disk in a set of points, and querying the size of the $i$'th convex layer in a planar set of points. We also answer a question of Chan et al.\ [SODA 2022] by giving a conditional lower bound for dynamic approximate square set cover. While many conditional lower bounds for dynamic data structures have been proven since the seminal work of P\u{a}tra\c{s}cu [STOC 2010], few of them relate to computational geometry problems. This is the first paper focusing on this topic. Most problems we consider can be solved in $O(n\log n)$ time in the static case and their dynamic versions have only been approached from the perspective of improving known upper bounds. One exception to this is Klee's measure problem in $\R^2$, for which Chan [CGTA 2010] gave an unconditional $\Omega(\sqrt{n})$ lower bound on the worst-case update time. By a similar approach, we show that such a lower bound also holds for an important special case of Klee's measure problem in $\R^3$ known as the Hypervolume Indicator problem, even for amortized runtime in the incremental setting.
\end{abstract}

\section{Introduction}

In 1995, Gajentaan and Overmars \cite{Gajentaan1995} introduced the notion of 3SUM hardness, showing that a number of problems in computational geometry can not be solved in subquadratic time, assuming the so-called 3SUM problem can not be solved in subquadratic time.\footnote{In 2014, Gr{\o}nlund and Pettie \cite{Gronlund2014} showed that the 3SUM problem can be solved in (slightly) subquadratic time. The modern formulation thus replaces ``subquadratic'' with ``truly subquadratic'', i.e.\  $O(n^{2-\epsilon})$ for some constant $\epsilon>0$.}
The general approach of proving polynomial lower bounds based on a few conjectures about key problems has since grown into its own subfield of complexity theory known as \emph{fine-grained complexity}. The most popular of these conjectures concern the aforementioned 3SUM problem, All-Pairs-Shortest-Paths (APSP), Boolean Matrix Multiplication (BMM), Triangle finding in a graph, Boolean Satisfiability (SAT) and the Orthogonal Vectors problem (2OV) (see for example the introductory surveys by Bringmann \cite{Bringmann2019} and V.~V.~Williams \cite{Williams2019}).  Another problem which crops up as a bottleneck in computational geometry is Hopcroft's problem (see the recent paper by Chan and Zheng~\cite{ChanHopcroft2022}).

P\u{a}tra\c{s}cu \cite{Patrascu2010} launched the study of such polynomial lower bounds for dynamic problems, where instead of simply computing a function on a single input, we want to be able to update that input and get the corresponding output of the function without having to recompute it from scratch. In particular, he introduced the Multiphase problem and showed a polynomial lower bound on its complexity, conditioned on the hardness of the 3SUM problem. Using the Multiphase problem as a stepping stone, he showed conditional hardness results for a variety of dynamic problems. Improvements and other conditional lower bounds for dynamic problems (data structure problems) have since appeared in the literature \cite{AbboudW2014, AmirCLL2014_finegrained, Henzinger2015, KarczmarzL2015_finegrained, AbboudD2016, Kopelowitz2016, KopelowitzK2016-finegrained, BaswanaCC02016-finegrained, Dahlgaard2016_finegrained, HenzingerL0W2017_finegrained, AlmanMW2017_finegrained, BerkholzKS2017-finegrained, BerkholzKS2018-finegrained, AbboudWY2018, Probst2018-finegrained, ChenDGWXY2018-finegrained, AmirKLPPS2019-finegrained, BrandNS2019-finegrained}.
Of particular interest for the purpose of this work is a paper by Kopelowitz et al.\ \cite{Kopelowitz2016} where the approach of P\u{a}tra\c{s}cu is improved by showing a tighter reduction from 3SUM to the so-called Set Disjointness problem (an intermediate problem between 3SUM and the Multiphase problem), as well as a paper by V.~V.~Williams and Xu \cite{Williams2020}, which obtains a similar reduction from the so-called Exact Triangle problem. Also particularly relevant here is the work of Henzinger et al.~\cite{Henzinger2015}, who show that many of the known bounds on dynamic problems can be derived (and even strengthened) by basing proofs on a hardness conjecture about the Online Boolean Matrix-Vector Multiplication (OMv) problem which they introduce.

While computational geometry was one of first fields where conditional lower bounds for algorithms were applied, for example by showing that determining if a point set is in general position is 3SUM hard \cite{Gajentaan1995}, the progress in conditional lower bounds for dynamic problems has not found widespread application to computational geometry; recent work has been largely confined to improved upper bounds. The only examples before the first version of this paper\footnote{We exclude from this list examples where (conditional) bounds on the static case trivially imply polynomial bounds on the dynamic case.} relate to (approximate) nearest-neighbor search under different metrics (see the paper by Rubinstein \cite{Rubinstein2018}, the introductory article by Bringmann \cite{Bringmann2021} as well as a preprint by Ko and Song \cite{Ko2020}), a paper by Lau and Ritossa \cite{LauR2021} with results for orthogonal range update on weighted point sets and an (unconditional) lower bound by Chan \cite{Chan2010} for a dynamic version of Klee's Measure Problem (see Sections \ref{subsec:rectangles} and \ref{subsec:hypervolume_indicator} of this paper). After a previous version of the present paper appeared on arXiv, and independent of our work, Jin and Xu \cite{Jin2022} studied generalized versions of the OMv and BMM problems and proved polynomial lower bounds for various dynamic problems based on their hardness, among which Dynamic 2D Orthogonal Range Color
Counting, Counting Maximal Points, Dynamic Klee's measure problem for unit hypercubes and Chan’s Halfspace Problem.

In this work, we exploit the results of P\u{a}tra\c{s}cu, Kopelowitz et al., V.~V.~Williams and Xu, and Henzinger et al.\ to give conditional polynomial lower bounds for a variety of dynamic problems in computational geometry, based on the hardness of 3SUM, APSP and Online Boolean Matrix-Vector Multiplication. Almost all the problems we study here share the common characteristic of being about computing a single global metric for a set of objects in space subject to updates. Moreover, in the static case (where there are no updates) most of these metrics can be computed in worst-case $O(n\log n)$ time using standard computational geometry results. In particular, we show conditional hardness results for orthogonal range marking, maintaining the number of maximal or extremal points in a set of points in $\R^3$, dynamic approximate square set cover, problems related to Klee's Measure Problem, problems related to finding the largest empty disk in a set of points, testing whether a set of disks covers a given rectangle, and querying for the size of the $i$'th convex layer of a set of points in the plane. 
We also give an unconditional lower bound for the incremental Hypervolume Indicator problem in $\R^3$, where the goal is to maintain the volume of the union of a set of axis-aligned boxes which all have the origin as one of their vertices.

The most basic of these problems, and the one we present first, is Square Range Marking: given a set of $n$ initially unmarked points in the plane, preprocess them to allow marking of the points in any given axis-aligned square and testing if there is any unmarked point. This encompasses the idea of augmenting a range query structure where augmentations can be applied to all data in a query range; a mark is the simplest such augmentation. While many variants of augmented orthogonal range queries have been studied (especially in the static case) \cite{Gupta2018_coloredSearching,JanardanL1993_coloredSearching, GuptaJS1995_coloredSearching,AgarwalGM2002_coloredSearching, Mortensen2003_coloredSearching,ShiJ2005_coloredSearching, LarsenW2013_coloredSearching,Nekrich2014_coloredSearching,ChanN2020_coloredSearching, ChanH2021_coloredSearching}, this natural variant has been given little attention. This is perhaps no coincidence, as we show that the straightforward polynomial-time solution based on kd-trees is likely almost optimal, in contrast to standard 1-D range marking and other augmentation problems which are easily handled by suitable variants of BSTs \cite[Ch.~14]{CLRS}.

Lau and Ritossa \cite{LauR2021} give similar lower bounds for data structures on weighted points, conditioned on the hardness of Online Boolean Matrix-Vector Multiplication, but explicitly leave open questions on points which have a color or a ``category.'' They show for example a lower bound for a data structure which allows to increment the weight of all points in an orthogonal range and to query the sum of weights for all points in a given range, as well as for variants of this problem.

\subsection{Setting and computational model}

We work in the standard Word RAM model, with words of $w=\Theta(\log n)$ bits unless otherwise stated, and for randomized algorithms we assume access to a perfect source of randomness. We will base our conditional lower bounds on the following well known hardness conjectures.
\begin{conjecture}
[3SUM conjecture]
The following problem (3SUM) requires $n^{2-o(1)}$ expected time to solve:
given a set of $n$ integers in $\{-n^3,\ldots n^3\}$, decide if three of them sum up to $0$.\footnote{The assumption that the integers are in $\{-n^3,\ldots n^3\}$ is done without loss of generality. In the model we consider one can always reduce the problem to this setting while preserving the expected run-time, via known hashing methods \cite{Baran2008}.}
\end{conjecture}

\begin{conjecture}[APSP conjecture]
The following problem (APSP) requires $n^{3-o(1)}$ expected time to solve:
given an integer-weighted directed graph $G$ on $n$ vertices with no negative cycles, compute the distance between every pair of vertices in $G$.
\end{conjecture}

The 3SUM problem can easily be solved in $O(n^2)$ time, while APSP can be solved in cubic time by the Floyd–Warshall algorithm, for example. The best known methods improve these runtimes by subpolynomial factors \cite{Chan2020-3SUM, Williams2018}.

In addition to being the basis for these standard conjectures in fine-grained complexity, the 3SUM problem and the APSP problem are related in other ways (see \cite{Williams2020}). In particular, they both fine-grained reduce to the Exact Triangle problem, meaning that if either the 3SUM conjecture or the APSP conjecture is true, then the following conjecture is true.
\begin{conjecture}[Exact Triangle conjecture]
The following problem (Exact Triangle) requires $n^{3-o(1)}$ expected time to solve:
given an integer-weighted graph $G$ and a target weight $T$, determine if there is a triangle in $G$ whose edge weights sum to $T$. 
\end{conjecture}
Thus, any bound conditioned on this conjecture also holds conditioned on the 3SUM conjecture or the APSP conjecture.

We also consider a conjecture introduced by Henzinger et al.\ \cite{Henzinger2015}, which can be thought of as a weakening of the informal conjecture which says that ``combinatorial'' matrix multiplication on $n\times n$ matrices requires essentially cubic time (note that the term ``combinatorial'' is not well defined).
\begin{problem}[Online Boolean Matrix-Vector Multiplication (OMv) \cite{Henzinger2015}]
We are given a $n\times n$ boolean matrix $M$. We can preprocess this matrix, after which we are given a sequence of $n$ boolean column-vectors of size $n$ denoted by $v_1,...,v_n$, one by one. After seeing each vector $v_i$, we must output the product $M v_i$ before seeing $v_{i+1}$.
\end{problem}
\begin{conjecture}[OMv conjecture]
Solving OMv requires $n^{3-o(1)}$ expected time in the worst case.
\end{conjecture}

The OMv problem can be solved in total time $O(n^3)$ by the naive algorithm. Here also the best known method improves this runtime by a subpolynomial factor \cite{Larsen2017}.

The conjecture was originally introduced in the Monte-Carlo setting (i.e.\ algorithms with a deterministic runtime but which are allowed to err with a small enough probability). We state it in the Las Vegas setting for the sake of uniformity of presentation. All the results of Henzinger et al.\ carry over to that setting with no difficulty.

While Henzinger et al.\ showed that most known lower bounds on dynamic problems derived from the 3SUM conjecture can be derived from the OMv conjecture (and often even strengthened), it is not known whether one conjecture implies the other. For most of our problems we derive polynomial lower bounds from both the OMv conjecture and the Exact Triangle conjecture. In such cases, we still get such lower bounds if at least one of the four considered conjectures is true. Moreover, the reductions used here could also give bounds in the case some of these conjecture fail by a small enough polynomial factor (for example if 3SUM requires $\Omega(n^{4/3})$ time).

Note also that recent work by Chan et al.\ \cite{ChanWX2022} directly implies that the lower bounds we obtain from the APSP conjecture also hold in the so-called Real RAM model (conditioned on the analogous Real-APSP conjecture) and in restricted versions of the model. For the real versions of the 3SUM and Exact Triangle conjectures, combining our reductions with theirs would also imply polynomial lower bounds for many of the problems considered here, although weaker than the ones we obtain in the Word RAM model.

\subsection{Main results}

\begin{table}
    \begin{center}
    \begin{tabular}{@{\hspace{-.5\arrayrulewidth}}l@{\hspace{-.5\arrayrulewidth}}ccc}
    \linenumbersep 27pt \socgnl
    &\textbf{Problem} & \textbf{Upper Bound} & \textbf{Lower Bound}\\
    
    \midrule
    
    \linenumbersep 27pt \socgnl
    && & 
                            \multirow{4}{*}{
                                \begin{tabular}{c}
                                    From Exact Triangle:\\
                                    $n^{1/4-o(1)\ \dag}$\medskip\\
                                    From OMv:\\
                                    $n^{1/2-o(1)\ \dag}$
                                \end{tabular}
                            }\\
                            
    \linenumbersep 27pt \socgnl
    &Square Range Marking [\S\ref{sec:square_range_marking}] &  \begin{tabular}{c}
                                $\tilde{O}(n^{1/2})^{\ \dag, \ddag}$ \cite{CardinalIK2021}
                            \end{tabular} &\\
                            
    \linenumbersep 27pt \socgnl
    && &\\
    
    \linenumbersep 27pt \socgnl
    && &\\
    \midrule
    
    \linenumbersep 27pt \socgnl
    &Counting Extremal Points in $\R^3$ [\S\ref{sec:maximal_extremal}]
            & 
            \multirow{4}{*}{
                \begin{tabular}{c}
                    $O^*(n^{7/8})^{\ \dag}$ \cite{Chan2003}\\ \\
                    $O^*(n^{11/12})^{\ \ddag}$ \cite{Chan2020}
                \end{tabular} 
            }
            &
            \multirow{6}{*}[-8pt]{
                \begin{tabular}{c}
                    From Exact Triangle:\\
                    $n^{1/5-o(1)\ \dag,\ddag}$\\
                    $n^{1/4-o(1)\ \ddag,\$}$\medskip\\
                    From OMv:\\
                    $n^{1/2-o(1)\ \dag,\ddag}$
                \end{tabular}
            } \\
            
    \cmidrule{2-2}
    
    \linenumbersep 27pt \socgnl
    &Largest Empty Disk in Query Region [\S\ref{subsec:disks}]& & \\
    
    \cmidrule{2-2}
    
    \linenumbersep 27pt \socgnl
    &Largest Empty Disk in a Set of Disks [\S\ref{subsec:disks}] &  & \\
    
    \cmidrule{2-2}
    
    \linenumbersep 27pt \socgnl
    &Rectangle Covering with Disks [\S\ref{subsec:rectangle_cover_with_disks}] &  & \\
    
    \cmidrule{2-3}
    
    \linenumbersep 27pt \socgnl
    &Square Covering with Squares [\S\ref{subsec:rectangles}] &         \begin{tabular}{c}
                                                $\tilde{O}(n^{1/2})^{\ \ddag}$ \cite{Yildiz2011-prob}
                                            \end{tabular} & \\
    \cmidrule{2-3}
    
    \linenumbersep 27pt \socgnl
    &Convex Layer Size in $\R^2$ [\S\ref{subsec:convex_layer_size}] & &\\
    
    \midrule
    
    \linenumbersep 27pt \socgnl
    && & 
            \multirow{5}{*}[-4pt]{
                \begin{tabular}{c}
                    From Exact Triangle:\\
                    $n^{1/4-o(1)\ \dag}$\\
                    $n^{1/3-o(1)\ \ddag}$ \medskip\\
                    From OMv:\\
                    $n^{1/2-o(1)\ \dag, \ddag}$
                \end{tabular}
            }\\
    
    \linenumbersep 27pt \socgnl
    &Counting Maximal Points in $\R^3$ [\S\ref{sec:maximal_extremal}] & \begin{tabular}{c}
                                            $\tilde{O}(n^{2/3})^{\ \ddag}$ \cite{Chan2020}
                                        \end{tabular} & \\
    \cmidrule{2-3}
    
    \linenumbersep 27pt \socgnl
    &$O(n^\alpha)$-approx.\ Weighted Square Set Cover [\S\ref{sec:square_set_cover}] &  & \\
    
    \cmidrule{2-3}
    
    \linenumbersep 27pt \socgnl
    &Klee's Measure Problem with Squares [\S\ref{subsec:rectangles}] &   \begin{tabular}{c} 
                                                $\tilde{O}(n^{1/2})^{\ \ddag}$ \cite{Yildiz2011-prob}
                                            \end{tabular} & \\
    
    \cmidrule{2-3}
    
    \linenumbersep 27pt \socgnl
    &Discrete KMP with Squares [\S\ref{subsec:rectangles}] &             \begin{tabular}{c}
                                                ${O}(n^{1/2})^{\ \dag,\ddag}$ \cite{Yildiz2011}
                                            \end{tabular} & \\
    
    \linenumbersep 27pt \socgnl                     && & \\
    
    \midrule
    
    \linenumbersep 27pt \socgnl 
    && &     \multirow{4}{*}{
                \begin{tabular}{c}
                    From Exact Triangle:\\
                    $n^{1/3-o(1)\ \dag, \ddag}$\medskip\\
                    From OMv:\\
                    $n^{1/2-o(1)\ \dag, \ddag}$
                \end{tabular}
            }\\
    
    \linenumbersep 27pt \socgnl         
    &Depth Problem with Squares [\S\ref{subsec:rectangles}] &    
            \begin{tabular}{c}
                $\tilde{O}(n^{1/2})^{\ \ddag}$ \cite{Yildiz2011-prob}
            \end{tabular} & \\
            
    \linenumbersep 27pt \socgnl 
    && &\\
    \linenumbersep 27pt \socgnl 
    && &\\
    
    \midrule 
    
    \linenumbersep 27pt \socgnl
    && & 
                            \multirow{3}{*}{
                                \begin{tabular}{c}
                                    From OMv:\\
                                    $n^{1/3-o(1)\ \dag, \ddag}$
                                \end{tabular}
                            }\\
                            
    \linenumbersep 27pt \socgnl
    &$O(1)$-approximate Square Set Cover [\S\ref{sec:square_set_cover}] &  \begin{tabular}{c}
                                $O^*(n^{1/2})^{\ \ddag}$\cite{Chan2022}
                            \end{tabular} &\\
                            
    \linenumbersep 27pt \socgnl
    && &\\

    \midrule
    
    \linenumbersep 27pt \socgnl 
    &Hypervolume Indicator in $\R^3$ [\S\ref{subsec:hypervolume_indicator}] &    
            \begin{tabular}{c}
                $\tilde{O}(n^{2/3})^{\ \ddag}$ \cite{Chan2020}
            \end{tabular} & 
                            \multirow{1}{*}{
                                \begin{tabular}{c}
                                    $\Omega(\sqrt{n})^{\ \#}$
                                \end{tabular}
                            }\\

    \bottomrule
\end{tabular}
\end{center}
    {\linenumbersep 11pt \socgnl 
    $^\dag$ per-operation runtime in the incremental setting.\\
    \socgnl 
    $^\ddag$ amortized runtime in the fully-dynamic setting.\\
    \socgnl 
    $^\$$ assuming $n^{1+o(1)}$ expected preprocessing time.\\
    \socgnl 
    $^\#$ unconditional lower bound in the incremental setting on amortized time, assuming at most polynomial time preprocessing, or on worst-case time without preprocessing assumptions.}\\
    
    \caption{Non-trivial known upper bounds and new (at the time of the first version of this paper being made public) lower bounds on the maximum over update and query time derived from the Exact Triangle conjecture, the OMv conjecture or (in the case of the Hypervolume Indicator problem) unconditionally. The $\tilde{O}$ notation hides polylog factors, while the $O^*$ notation hides factors which are $o(n^\epsilon)$ for an arbitrarily small constant $\epsilon>0$.  All upper bounds are for data structures with at most $O^*(n)$ preprocessing. Note that the lower bounds for Square Range Marking also hold in the case of a static set of points (with some assumptions on preprocessing time) and that the lower bound for the Depth Problem derived from the OMv conjecture also holds for amortized runtime in the incremental setting. The lower bound obtained for counting maximal points has since been superseded by the more general result of Jin and Xu \cite{Jin2022} who obtain lower bounds also in higher dimension.}
    \label{tab:main-bounds}
\end{table}

We obtain (conditional) polynomial lower bounds for a variety of dynamic geometric problems, and an unconditional bound for the incremental Hypervolume Indicator problem in $\R^3$. Our bounds are stated as inequalities which imply trade-offs between achievable update and query times. The lower bounds we get on the maximum of both are summarized in Table \ref{tab:main-bounds}, together with known upper bounds. Note that the bounds we get for squares or square ranges imply the same bounds for rectangles or general orthogonal ranges, although we sometimes get better trade-offs in these cases. 

Some of the lower bounds reveal interesting separations between geometric dynamic problems whose operations can be supported in subpolynomial or $O(n^\epsilon)$ time and similar problems which require polynomial time with a fixed exponent (under the hardness conjectures we consider).
\begin{itemize}
    \item Orthogonal range queries with dynamic updates on single points can be done with $\ppolylog$ time operations, while dynamic updates on orthogonal ranges of points require polynomial time.
    \item Dynamically maintaining maximal points in a point set can be done in $\ppolylog$ time in $\R^2$, while maintaining only their number in $\R^3$ already requires polynomial time.
    \item The same separation between dimensions $2$ and $3$ applies for maintaining (the number of) extremal points.
    \item Related to the previous point, the ability to query for the size of any convex layer on a dynamic set of points in $\R^2$ requires polynomial time (compared to $\ppolylog$ time when we are only interested in the first convex layer, i.e.\ the convex hull).
    \item Maintaining a $O(1)$-approximation for the size of dynamic unit square set cover can be done in $2^{O(\sqrt{\log n})}$ amortized time per update \cite{Chan2022}, while maintaining the size of a $O(n^\alpha)$-approximation (for a constant $0\leq \alpha<1$) requires polynomial time for arbitrarily sized squares (with an exponent dependent on $\alpha$).
    \item In the weighted case of the previous problem, we also get such a separation: $O(1)$-approximate weighted unit square set cover can be done in $O(n^\epsilon)$ time \cite{Chan2022} while $O(n^\alpha)$-approximate weighted dynamic square set cover requires polynomial time, with an exponent independent of $\alpha$.
\end{itemize}

\section{The general approach}
In all the problems we consider, we have a data structure $D$ which maintains a set $S$ of $O(n)$ geometric objects, supporting some form of update and query (a query is any operation which never impacts the result of any subsequent operation). We say that a data structure (or that the set of objects it maintains) is \emph{incremental} when it allows updates which consist of inserting a new object in $S$. We use the term \emph{fully-dynamic} when both insertions and deletions are allowed. The set $S$ can be initialized in a preprocessing phase.

\subsection{General reduction schemes}\label{sec:general_reductions}

All our reductions have the same basic structure based on a geometric view of P\u{a}tra\c{s}cu's Multiphase problem \cite{Patrascu2010}, where we encode a family $\mathcal{F}=\{F_1,\ldots F_k\}$ of subsets of $\{1,\ldots, m\}$  as a grid of objects where the presence (or absence) of an object at the grid coordinates $(x,y)$ encodes $x\in F_y$. We can then select some of the columns $I\in \{1,\ldots,k\}$ and a row $j\in\{1,\ldots, m\}$, allowing us to test if $I\cap F_j \neq \emptyset$ efficiently. We abstract some of the commonalities of the reductions in the following ``general'' reduction schemes, so we can focus on the specifics of each problem and avoid repetitions later on. Rather than give the original definition of the Multiphase problem, let us define what it means for a data structure to solve it, as this will make the statements of reductions easier, more uniform, and makes the required constraints on the data structure we consider explicit.
\begin{definition}[Solving the Multiphase problem]
Let $\mathcal{F}=\{F_1, \ldots, F_k\}$ be a family of $k$ subsets of $\{1,2,\ldots m\}$. Let $s_\mathcal{F} = \sum_{F\in\mathcal{F}}|F|$.
Consider a data structure $D$ with an undo operation\footnote{A data structure is said to have an undo operation if for any update $U$ there is complementary update $U'$ so that if $U$ and $U'$ are executed sequentially the results of subsequent operations are identical to the case where $U$ and $U'$ were never executed. This requirement is easily satisfied in structures that maintain a set and have insertion and deletion update operations.} which maintains a set $S$ of $O(n)$ objects with expected preprocessing time $O(t_p)$, expected amortized update time $O(t_u)$ and expected amortized query time $O(t_q)$. Suppose it allows us to do the following.
\begin{itemize}
    \item (Step 1) First, we read $\mathcal{F}$ and store a set of $n$ objects in $S$ using only the preprocessing operation of $D$.
    \item (Step 2) Then, we receive a subset $J \subset \{1,2,\ldots m\}$ and perform $u_J$ updates on $S$.
    \item (Step 3) Finally, we are given an index $1\leq i\leq k$ and after $O(1)$ updates and queries on $S$ we decide if $J \cap F_{i} \neq \emptyset$.
\end{itemize}
Assume that the time of each of these three steps is dominated by the time of the operations on $D$ and that in each step, the only information available from the previous steps is what is accessible through $D$.
Let $t_{uq} = t_q$ if only queries are performed in Step 3, otherwise let $t_{uq} = t_u+t_q$.

We say that such a data structure solves the Multiphase problem.
\end{definition}
As mentioned in the introduction, P\u{a}tra\c{s}cu gave lower bounds on the time required to solve the Multiphase problem conditioned on the 3SUM conjecture and reduced this problem to various dynamic problems. His reduction from 3SUM has since been tightened by Kopelowitz et al.\ \cite{Kopelowitz2016} and reductions from the Exact Triangle and OMv conjectures have been found by Vassilevska Williams and Xu \cite{Williams2020} and Henzinger et al.\ \cite{Henzinger2015} respectively.

We summarize the implications from these works for different parameters in the following theorems. While this results in somewhat verbose statements, we chose this approach in order to streamline the reductions in this paper and to make the lower bounds we obtain explicit in terms $n$.
\begin{theorem}\label{thm:general_3sum}
Let $D$ be a data structure which solves the Multiphase problem.
If the Exact Triangle conjecture is true (or in particular if either the 3SUM or APSP conjecture is true), then for any $0<\gamma<1$:
\begin{itemize}
    \item{} (Scenario 1) If $n = O(m\cdot k)$ and $u_J = O(m)$, we have
    \[t_p + t_u\cdot n + t_{uq}\cdot n^\frac{1+\gamma}{3-2\gamma} = \Omega\left(n^{\frac{2}{3-2\gamma}-o(1)}\right).\]
    
    \item{} (Scenario 2) If $n = O(m\cdot k)$ and $u_J = O(|J|)$, we have 
     \[t_p + t_u\cdot n^\frac{2-\gamma}{3-2\gamma} + t_{uq}\cdot n^\frac{1+\gamma}{3-2\gamma}  
     = \Omega\left(n^{\frac{2}{3-2\gamma}-o(1)}\right).\]
     
    \item{} (Scenario 3) If $n = O(s_\mathcal{F})$ and $u_J = O(m)$, we have 
    \[t_p + t_u\cdot n^\frac{3-2\gamma}{2-\gamma} + t_{uq}\cdot n^\frac{1+\gamma}{2-\gamma} 
     = \Omega\left(n^{\frac{2}{2-\gamma}-o(1)}\right).\]
    
    \item{} (Scenario 4) If $n = O(s_\mathcal{F})$ and $u_J = O(|J|)$, we have
    \[t_p + t_u\cdot n + t_{uq}\cdot n^\frac{1+\gamma}{2-\gamma} 
     = \Omega\left(n^{\frac{2}{2-\gamma}-o(1)}\right).\]
\end{itemize}
\end{theorem}
Note that for incremental (or fully-dynamic) data structures where we can insert objects, we can always assume $t_p = O(t_u\cdot n)$ by inserting the $O(n)$ initial objects individually.

The results of Henzinger et al.\ \cite{Henzinger2015} imply that whenever we have such lower bounds from the hardness of Exact Triangle, we can get stronger bounds if we assume hardness of the OMv problem instead.

\begin{theorem}\label{thm:general_omv}
Let $D$ be a data structure which solves the Multiphase problem.
Assume $n = O\left(m^{c_1}\cdot k^{c_2}\right)$ for some constants $c_1, c_2 > 0$, $u_J=O(m)$, and the expected preprocessing time $t_p$ is at most polynomial in $n$. If the OMv conjecture is true, then for any $0<\gamma<1$,
\[t_u\cdot n^{\gamma} + t_{uq} \cdot n^{\frac{1-c_1\gamma}{c_2}} = \Omega\left(n^{\frac{1+(c_2-c_1)\gamma}{c_2}-o(1)}\right).\]

In particular if $n=O(m\cdot k)$ (as is the case in the four scenarios of Theorem \ref{thm:general_3sum}), for any $0<\gamma<1$ we have
\[t_u\cdot n^{\gamma} + t_{uq}\cdot n^{1-\gamma} 
 = \Omega\left(n^{1-o(1)}\right).\]
For $\gamma = 1/2$, we thus have
$t_u + t_q  = \Omega\left(n^{1/2-o(1)}\right)$.
\end{theorem}

In some cases, the methods of Henzinger et al.\ \cite{Henzinger2015} allow us to get a better bound on the relation between update and query time. We will use this in the following form.
\begin{theorem}\label{thm:general_oumv}
Let $\mathcal{F}=\{F_1, \ldots, F_N\}$ be a family of $N$ subsets of $\{1,2,\ldots N\}$. Let $s_\mathcal{F} = \sum_{F\in\mathcal{F}}|F|$.
Consider a data structure $D$ with an undo operation which maintains a set $S$ of $O(n)$ objects with expected preprocessing time $O(t_p)$, expected amortized update time $O(t_u)$ and expected amortized query time $O(t_q)$. Suppose it allows us to do the following:
\begin{itemize}
    \item (Step 1) First, we read $\mathcal{F}$ and store a set of $n= O(N^2)$ objects in $S$ in a preprocessing phase.
    \item (Step 2) Then, we receive two subsets $I,J \subset \{1,2,\ldots N\}$ and perform $O(N)$ updates on $S$.
    \item (Step 3) Finally, after a constant number of queries, we decide if there exists $i\in I$ and $j\in J$ such that $j \in F_{i}$.
\end{itemize}
We assume that the time of each of these three steps is dominated by the time of the operations on $D$.
If the OMv conjecture is true and and the preprocessing time $t_p$ is at most polynomial then
$t_u\cdot \sqrt{n} + t_q  = \Omega\left(n^{1-o(1)}\right)$.
In other words, either $t_u = \Omega\left(n^{1/2-o(1)}\right)$ or $t_q = \Omega\left(n^{1-o(1)}\right)$.
\end{theorem}

All these results follow from straightforward adaptations of P\u{a}tra\c{s}cu's proofs \cite{Patrascu2010} together with the more recent results from Williams and Xu \cite{Williams2020} and Henzinger et al.\ \cite{Henzinger2015}, and are implicit in the two latter papers. 
To apply these theorems, we need data structures with an undo operation. When considering structures in the fully-dynamic setting where updates consist of inserting or deleting an object, then this requirement is automatically satisfied. For structures with guarantees on the runtime per operation (rather than amortized guarantees), we can use the following standard technique (see for example \cite[Theorem 2.1]{OvermarsPast1981}).

\begin{lemma}
Any data structure with guarantees on the runtime per operation (non-amortized) can be augmented to support an undo operation with the same guarantees.
\end{lemma}

From now on, whenever we consider a structure with per-operation runtime guarantees, we assume (without loss of generality) that it has been augmented to support undo.

\subsection{An example: Square Range Marking}\label{sec:square_range_marking}

We illustrate the use of these theorems on the following problem.

\subparagraph*{Square Range Marking} Preprocess a static set of $n$ initially unmarked points, where an update consists of marking all points in a given axis-aligned square range and a query returns if there is any unmarked point in the set.

\medskip
Here the dynamic part of the problem is rather limited as only the markings of the points can change after an update, the set of points itself is static. The updates are even monotone in the sense that once a point has been marked it is never unmarked (in particular, the number of unmarked points can never increase). Even for this seemingly simple problem, we can use Theorems \ref{thm:general_3sum} and \ref{thm:general_omv} to get the following (conditional) polynomial lower bounds.
\begin{theorem}
Let $D$ be a data structure for Square Range Marking with $t_p$ expected preprocessing time and $t_u$ expected time per update (i.e.\  non-amortized). If the Exact Triangle conjecture holds, then
\[t_p + t_u\cdot(n^\frac{3-2\gamma}{2-\gamma} + n^\frac{1+\gamma}{2-\gamma}) + t_q \cdot n^\frac{1+\gamma}{2-\gamma} = \Omega\left( n^{\frac{2}{2-\gamma}-o(1)}\right).\]
If the OMv conjecture holds and $t_p$ is at most polynomial then for any $0<\gamma<1$
\[t_u\cdot (n^{1-\gamma} + n^{\gamma}) + t_{q}\cdot n^{1-\gamma} = \Omega\left(n^{1-o(1)}\right).\]
In particular, by setting $\gamma = 1/2$, we have
$t_u + t_q = \Omega\left(n^{1/2 - o(1)}\right)$.
\end{theorem}
\begin{proof}
It suffices to show that such a data structure fits the conditions of Scenario 3 in Theorem \ref{thm:general_3sum}. Let $\mathcal{F}=\{F_1, \ldots, F_k\}$ be a family of $k$ subsets of $\{1,2,\ldots m\}$. 

We perform Step 1 by initializing $D$ with the following points: for each $1\leq i\leq k$ and $1\leq j \leq m$ for which $j \in F_i$, we put a point $p_{i,j}$ at coordinates $((k+2)j+1,i+1)$.
The total number of points is $n = s_\mathcal{F}$. 

To perform Step 2 when given $J \subset \{1,2,\ldots m\}$, we mark the points inside a square range of side-length $k+2$ whose lower-left corner has coordinates $((k+2)j,1)$, for all $j \not\in J$. This requires $O(m)$ updates on $D$. The unmarked points are exactly the $p_{i,j}$'s such that $j\in J$.

In Step 3, when given an index $1\leq i'\leq k$, we mark the points inside the two squares of side-length $(k+2)m$ whose lower-left corners lie at coordinates $(0,i'+ 1/2)$ and $(0,i'-(k+2)m-1/2)$ respectively.
Now there is an unmarked point if and only if there is some point $p_{i,j}$ such that $j\in J$ and $p_{i,j}$ was not marked by these two last updates. This is the case if and only if $i=i'$. By construction, such a point exists if and only if there is some $j\in J$ such that $j\in F_{i'}$ (i.e.\  $J\cap F_{i'} \neq \emptyset$). Thus, we can answer a Step 3 query after two more updates to $D$.

By applying Theorems \ref{thm:general_3sum} and \ref{thm:general_omv} we get the result.
\end{proof}

If we assume truly subquadratic expected preprocessing time we get polynomial lower bounds on $t_u$ from the Exact Triangle conjecture. In particular we have the following.

\begin{corollary}
Let $D$ be a data structure for Square Range Marking with $t_p$ expected preprocessing time, $t_u$ expected time per update (i.e.\  non-amortized) and $t_q$ expected time per query. If the Exact Triangle conjecture is true and $t_p = o(n^{3/2})$ then
$t_u + t_q = \Omega\left(n^{1/4-o(1)}\right)$.
\end{corollary}

Consider now the analogous problem with arbitrary axis-aligned rectangles instead of squares (which we will call Rectangle Range Covering). For this problem, we can get a stronger condition on the relation between update and query time using Theorem \ref{thm:general_oumv}.

\begin{theorem}
Let $D$ be a data structure for Rectangle Range Marking with $t_p$ expected preprocessing time, $t_u$ expected time per update (i.e.\  non-amortized) and $t_q$ expected time per query. If the OMv conjecture is true and $t_p$ is at most polynomial then
\[t_u\cdot \sqrt{n} + t_q = \Omega\left( n^{1-o(1)}\right).\]
\end{theorem}

The result essentially states that if we want something faster than the trivial linear query time, we need almost $\sqrt{n}$ update time. This bound is most relevant in situations where updates happen more frequently than queries and we might have hoped to obtain $O(n^{1-\epsilon})$ query time together with very fast update times.

\begin{proof}
We show that such a data structure fits the conditions of Theorem \ref{thm:general_oumv}. Let $\mathcal{F}=\{F_1, \ldots, F_N\}$ be a family of $N$ subsets of $\{1,2,\ldots N\}$.

We perform Step 1 by initializing $D$ as follows: for each $1\leq i\leq N$ and $1\leq j \leq N$ for which $j \in F_i$, we put a point $p_{i,j}$ at coordinates $(j,i)$.
The number of points is $n = O(N^2)$. 

To perform Step 2 when given subsets  $I,J \subset \{1,2,\ldots N\}$, we mark the points $p_{i,j}$, for all $j \not\in J$ and all $i \not\in J$. This can easily be done using $O(m)$ updates on $D$. 

Now the unmarked points are exactly the points $p_{i,j}$ such that $i\in I$ and $j\in J$. By construction, such a point exists if and only if there is $i\in I$ and $j\in J$ such that $j\in F_{i}$. Thus, we can answer such a Step 3 query after a single query to $D$.

By applying Theorem \ref{thm:general_oumv} we get the result.
\end{proof}

The bounds on $t_u+t_q$ obtained for Square (or Rectangle) Range Marking from the OMv conjecture are almost tight, as an upper bound can easily be obtained by taking a two-dimensional kd-tree \cite{Bentley2015} and augmenting it by adding markers to the nodes indicating if the points in the corresponding subtrees are marked. We then get a data structure with $O(n\log n)$ worst-case preprocessing time and $O(\sqrt{n})$ worst-case time per update. As noted by Cardinal et al.\ \cite{CardinalIK2021}, using standard dynamization techniques such a data structure can even be made to support insertion and deletion of points in $O(\log^2 n)$ worst-case time.

\section{Counting the number of maximal or extremal points in \texorpdfstring{$\R^3$}{R3}}\label{sec:maximal_extremal}

We show polynomial bounds for the problems of counting the number of maximal or extremal points in a dynamic set $S$ of points in $\R^3$. Recall that a point $(x,y,z) \in S$ is maximal if there is no other distinct point $(x',y', z')\in S$ with $x'\geq x$, $y'\geq y$ and $z'\geq z$. A point is extremal if there is a plane passing through that point such that all other points lie strictly on one side of the plane. Note that in the plane these problems can be solved in polylog worst-case time per operation by known techniques \cite{Overmars1981}. There is thus a clear separation here between dimensions $2$ and $3$.

After a previous version of the present paper appeared on arXiv, Jin and Xu \cite{Jin2022} independently gave lower bounds for counting maximal points in odd dimensional space, conditioned on a generalization of the OMv conjecture. The lower bound we get here from the OMv problem is a special case of their result for dimension $3$.

\subparagraph*{Counting Maximal Points in $\R^3$} Maintain a dynamic set of $O(n)$ points in $\R^3$ and allow queries that return the number of maximal points in the set.

\begin{theorem}\label{thm:maximal_pts_ds}
Let $D$ be a fully-dynamic data structure for Counting Maximal Points in $\R^3$ with $t_p$ expected preprocessing time, $t_u$ expected amortized update time and $t_q$ expected amortized query time. We can assume  $t_p = O(t_u\cdot n)$. If the Exact Triangle conjecture holds, then
\[t_u\cdot(n + n^\frac{1+\gamma}{2-\gamma})+t_q\cdot n^\frac{1+\gamma}{2-\gamma} = \Omega\left( n^{\frac{2}{2-\gamma}-o(1)}\right).\]
If the OMv conjecture holds, then for any $0 < \gamma < 1$
\[t_u\cdot (n^{1 - \gamma} + n^{\gamma}) + t_q \cdot n^{1-\gamma} = \Omega\left(n^{1-o(1)}\right).\]
In particular, for $\gamma = 1/2$ we have
$t_u + t_q = \Omega\left( n^{1/2 - o(1)}\right)$. 
\end{theorem}

\begin{proof}
It suffices to show that such a data structure fits the conditions of Scenario 4 in Theorem \ref{thm:general_3sum}. Let $\mathcal{F}=\{F_1, \ldots, F_k\}$ be a family of $k$ subsets of $\{1,2,\ldots m\}$. Suppose without loss of generality that no set in $\mathcal{F}$ is empty and every element in $\{1,2,\ldots m\}$ appears in at least one set.
We perform Step 1 by initializing $D$ as follows.
\begin{itemize}
    \item For all $1\leq i \leq k$ and all $j\in F_i$, we put a point $p_{i,j}$ at coordinates $(j,i,-(k+2)j-i)$.
    \item For all $1\leq j \leq m$ we put a point $b_j$ at coordinate $(j, k+1, -(k+2)j+1)$.
\end{itemize}
This costs $t_u$ expected time, for a total number of points $n=\Theta(s_\mathcal{F})$.

Note that without the points $b_\bullet$, all other points would start out as maximal. The purpose of the point $b_j$ is to prevent exactly all points of the form $p_{\bullet, j}$ from being maximal.

To perform Step 2 when given $J \subset \{1,2,\ldots m\}$, we delete the points $b_{j}$ for all $j\in J$. This requires $O(|J|)$ updates on $D$.

In Step 3, we are given some index $1\leq i' \leq k$ and want to know if the set $F_{i'}$ intersects $J$. Because at this step the only maximal points of the form $p_{i',j}$ are those for which $j\in J$, this is equivalent to asking whether there is such a maximal point $p_{i',\bullet}$.
We start by inserting a point $t$ at coordinates $(m+1, i'-1/2, 1)$, which dominates exactly all points of the form $p_{i,\bullet}$ for $i<i'$. Then we get the number $c$ of maximal points from the data structure. We again insert a point $t'$ at coordinates $(m+2, i'+1/2, 2)$, which dominates exactly $t$ and all points of the form $p_{i,\bullet}$ for $i<i'+1$. We get the new number $c'$ of maximal points. We have $c'<c$ if and only if at least one maximal point of the form $p_{i',\bullet}$ got lost between the first and the second count, that is, if and only if $S_{i'}$ intersects $J$. Thus, we can get the answer to this intersection query after a constant number of updates to the data structure in Step 3. 

By applying Theorem \ref{thm:general_3sum} and Theorem \ref{thm:general_omv} we then get the result.
\end{proof}

For purely incremental data structures with per-update runtime guarantees (rather than amortized), we can adapt the proof of Theorem \ref{thm:maximal_pts_ds} so that in Step 1 we do not insert the points $b_\bullet$ and in Step 2 we insert the points $b_{j}$ for all $j \not\in J$ (instead of deleting those for which $j \in J$). This leads to the same bounds from the OMv problem and slightly worse lower bounds from the Exact Triangle problem (corresponding to Scenario 3 in Theorem \ref{thm:general_3sum}).

Because we can assume $t_p = O(t_u\cdot n)$ in both cases, we get the following corollaries:
\begin{corollary}
Let $D$ be a fully-dynamic data structure for Counting Maximal Points in $\R^3$ with $t_u$ expected amortized update time and $t_q$ amortized query time. If the Exact Triangle conjecture is true then
$t_u + t_q = \Omega\left(n^{1/3-o(1)}\right)$.
\end{corollary}

\begin{corollary}
Let $D$ be an incremental data structure for Counting Maximal Points in $\R^3$ with $t_u$ expected time per update. If the Exact Triangle conjecture is true then
\[t_u + t_q = \Omega\left( n^{1/4-o(1)}\right).\]
If the OMv conjecture is true then
\[t_u + t_q = \Omega\left(n^{1/2-o(1)}\right).\]
\end{corollary}

In the fully-dynamic setting, Chan \cite{Chan2020} gives a data structure for this problem with $O(n\cdot\ppolylog n)$ preprocessing time and $O(n^{2/3}\ppolylog n)$ amortized update and query time.

\subparagraph*{Counting Extremal Points in $\R^3$} Maintain a dynamic set of $O(n)$ points in $\R^3$ and allow for queries  counting the number of extremal points in the set.

\begin{theorem}\label{thm:extremal_pts_ds}
Let $D$ be a fully-dynamic data structure for Extremal Points in $\R^3$ with $t_p$ expected preprocessing time, $t_u$ expected amortized update time and $t_q$ expected amortized query time. If the Exact Triangle conjecture holds, then
\[t_p + t_u \cdot \left(n^{\frac{2-\gamma}{3-2\gamma}}+n^{\frac{1+\gamma}{3-2\gamma}}\right) + t_q \cdot n^{\frac{1+\gamma}{3-2\gamma}} = 
\Omega\left(n^{\frac{2}{3-2\gamma}-o(1)}\right).\]
Because we can assume $t_p = O(t_u\cdot n)$, if the OMv conjecture holds, then for any $0 < \gamma < 1$
\[t_u\cdot (n^{1 - \gamma} + n^{\gamma}) + t_q \cdot n^{1-\gamma} = \Omega\left(n^{1-o(1)}\right).\]
In particular, by setting $\gamma = 1/2$, we have
$t_u + t_q = \Omega\left(n^{1/2 - o(1)}\right)$.
\end{theorem}

\begin{figure}
    \centering
    \includegraphics[scale=0.63]{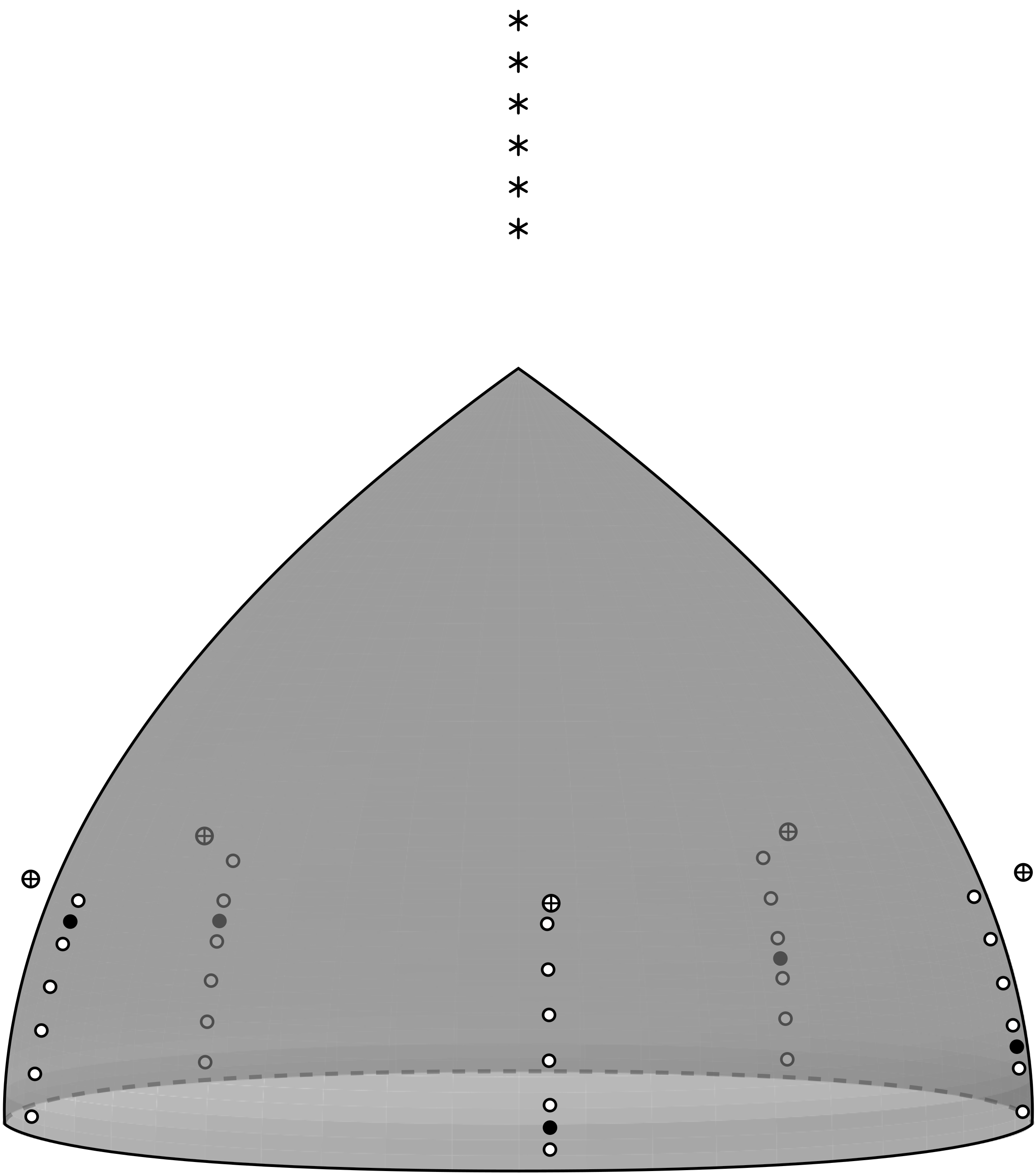}
    \caption{Illustration (not to scale) of the set of points obtained for $m=k=5$ and the family of sets $\mathcal{F} = \{\{1\}, \{2\}, \ldots, \{5\} \}$. The points $q_{\bullet, \bullet}$ are represented in white, the points $p_{\bullet, \bullet}$ in black and all these points lie on the translucent gray surface. The points $b_{\bullet}$ are represented with a cross in a white circle and are above the translucent gray surface. The asterisks represent the points of the form $t_\bullet$ and lie on the axis of rotational symmetry of the translucent gray surface.}
    \label{fig:extremal}
\end{figure}

Before proving this theorem, let us introduce some notation. We let $\mathcal{F} = \{F_1, \ldots, F_k\}$ be a family of $k$ subsets of $\{1,\ldots m\}$, where $m,k\geq 5$. Here we work in cylindrical coordinates $(r,\theta, z)$ (where this would denote the point $(r\cos\theta, r\sin\theta, z)$ in Cartesian coordinates).

Let $R = 4k^2m^2$. For all $0 \leq i \leq k$ and $1\leq j \leq m$, let $q_{i,j}$ be the point with cylindrical coordinates $(R-(2i+1)^2, \frac{2\pi}{m}j, 2i+1)$. Similarly, for all $1\leq j \leq m$ and $1 \leq i \leq k$ such that $j \in F_i$, let $p_{i,j}$ be the point with cylindrical coordinates $(R-(2i)^2, \frac{2\pi}{m}j, 2i)$. We let $S_\mathcal{F}$ denote the set consisting of all these points.

For all $1\leq j \leq m$, we let $b_j$ denote the point with cylindrical coordinates $(R-1, \frac{2\pi}{m}j, 2k+2)$. For all $1\leq i \leq k+1$, we let $t_i$ denote the point lying on the longitudinal axis (the $z$-axis) at height $\frac{R+(2i-1)^2}{2(2i-1)}$. Note that for all $1\leq i_1 \leq i_2 \leq k+1$, the point $t_{i_1}$ is higher on the $z$-axis than $t_{i_2}$ and that all points $t_\bullet$ have a larger $z$-coordinate than all other previously defined points.
See Figure \ref{fig:extremal} for an illustration.

\begin{lemma}\label{lemma:extremal}
Let $S$ be a set of points such that $S_\mathcal{F} \subset S \subset S_\mathcal{F} \cup \{b_j\mid 1\leq j \leq m\} \cup \{t_i\mid 1\leq i \leq k\}$. Let $1\leq j' \leq m$ and $1\leq i' \leq k$.
Then:
\begin{itemize}
    \item The point $q_{0,j'}$ is extremal.
    \item The point $q_{i',j'}$ is extremal if and only if $b_{j'} \not\in S$ and for all $1\leq i \leq i'$, $t_{i} \not\in S$. If $p_{i',j'} \in S$ (i.e.\ $j' \in F_{i'}$), then the same holds for $p_{i',j'}$.
    \item If $t_{i'}\in S$, then $t_{i'}$ is extremal if and only if for all $1\leq i < i'$, $t_{i}\not\in S$.
\end{itemize}

Moreover, this remains true even if all points are arbitrarily perturbed by moving them a distance of at most $1/R^2$.
\end{lemma}
\begin{proof}
By rotational symmetry considerations around the $z$-axis, it is enough to consider the case $j'=m$. To show that some point is extremal, it is enough to show that it is extremal in the set of points obtained by projecting $S$ orthogonally on the plane of Cartesian equation $y=0$. We let $S'$ denote this planar point set and identify the points in $S'$ with the corresponding points in $S$.

To show that $q_{0,m}$ is extremal in $S'$ (and thus in $S$), it suffices to notice that line $\ell_1$ of slope $1$ passing through $q_{0,m}$ has all other points of $S'$ lying strictly above it. Moreover, every other point is at distance at least $2\sqrt{2}$ from $\ell_1$, so this remains true even under a perturbation of the point set of magnitude at most $1/R^2 < \sqrt{2}$.

If for all $1\leq i \leq i'$, $t_{i}\not\in S'$, then the horizontal line $\ell_2$ passing through $t_{i'}$ has all other points of $S'$ lying strictly below it. Moreover, every other point is at distance greater than $2$ from $\ell_2$, so this remains true even under a perturbation of magnitude $1/R^2 < 1$. However, if there is some $1\leq i < i'$ such that $t_i \in S$, then the convex hull $H_1$ of the points $\{t_i, q_{0,1}, q_{0,2}, \ldots q_{0,m}\} \subset S$ contains $t_{i'}$ in its interior. Moreover, the distance from $t_{i'}$ to any face of $H_1$ at least \begin{align*}
&(R-1)\cos\left(\frac{\pi}{m}\right)
\frac{
\frac{R+(2i-1)^2}{2(2i-1)} - \frac{R+(2i'-1)^2}{2(2i'-1)}
}{
\sqrt{\left(\frac{R+(2i-3)(2j-1)}{2(2i-1)}\right)^2+(R-1)^2 \cos^2(\frac{\pi}{m})}
} \\
& \geq (R-1)\cos\left(\frac{\pi}{5}\right)
\frac{
1
}{
\sqrt{\left(\frac{R}{2}\right)^2+(R-1)^2}
} \\
& \geq (R-1)\cos\left(\frac{\pi}{5}\right)
\frac{
1
}{
\sqrt{2}(R-1)
}\\
&> 2/R^2.
\end{align*}
Thus, $t_{i'}$ remains inside $H_1$ even under an arbitrary perturbation of the points of magnitude at most $1/R^2$.

Now suppose that $b_{m} \not\in S'$, $t_{i} \in S'$ for all $i' < i \leq k$ and for all $1\leq i \leq i'$ $t_{i} \not\in S'$. Consider the line $\ell_3$ passing through $q_{i',m}$ and $t_{i'+1}$. One can show that all other points lie strictly below $\ell_3$, so $q_{i',m}$ is extremal. Moreover, it is easy to see that the point in $S'$ closest to $\ell$ (apart from $q_{i',m}$ and $t_{i'+1}$) is necessarily one among $p_{i',m}$, $p_{i'+1,m}$, $t_{i'+2}$, $b_{m-1}$ or $b_1$. One can check that all these points are at distance greater than $2/R^2$ from $\ell_3$. Thus, $q_{i',m}$ remains extremal even under an arbitrary perturbation of the points of magnitude at most $1/R^2$.

On the other hand, if $b_m \in S$, then the convex hull $H_2$ of the points $\{b_m, q_{0,m},q_{k, m-1}, q_{k,1}\}$ contains the points $q_{i',m}$ and $p_{i',m}$ in its interior. Here again a mechanical (but tedious) computation shows that these two points are at distance greater than $2/R^2$ to all faces of $H_2$.

Similarly, if $t_{i} \in S$ for some $1\leq i \leq i'$, then $q_{i',m}$ and $p_{i',m}$ are contained in the convex hull $H_3$ of $\{t_{i}, q_{i-1,1}, q_{i-1,2}, \ldots q_{i-1,m}\}$ and are at distance greater than $2/R^2$ to all faces of $H_3$.
\end{proof}

\begin{proof}[Proof of Theorem \ref{thm:extremal_pts_ds}]
It suffices to show that such a data structure fits the conditions of Scenario 2 in Theorem \ref{thm:general_3sum}. Let $\mathcal{F}=\{F_1, \ldots, F_k\}$ be a family of $k$ subsets of $\{1,2,\ldots m\}$. We use the notation of Lemma \ref{lemma:extremal}. We first describe the procedure without discussing issues of finite precision and later show how this can be carried out on a Word RAM machine with words of $O(\log n)$ bits.

We perform Step 1 by initializing $D$ with all points of the form $q_{\bullet,\bullet}$, $p_{\bullet,\bullet}$ and $b_\bullet$. This costs $t_u$ expected time, for a total number of points $n=\Theta(m\cdot k)$.

To perform Step 2 when given $J \subset \{1,2,\ldots m\}$, we delete the points $b_{j}$ for all $j\in J$. This requires $O(|J|)$ updates on $D$.

In Step 3, when given an index $1\leq i' \leq k$ we start by inserting the point $t_{i'+1}$ to $D$ and getting the count $c$ of extremal points. By Lemma \ref{lemma:extremal}, the extremal points of $S$ at this point are exactly those of the following $6$ types:
\begin{enumerate}
    \item the point $t_{i'+1}$,
    \item the points $b_j$ for all $j$ such that $j \not\in J$,
    \item the points $q_{i,j}$ for all $i,j$ such that $j \in J$ and $i < i'$,
    \item the points $p_{i,j}$ for all $i,j$ such that $j\in F_i$, $j \in J$ and $i < i'$.
    \item the points $q_{i',j}$ for all $j$ such that $j\in J$,
    \item the points $p_{i',j}$ for all $j$ such that $j\in F_{i'}$ and $j \in J$.
\end{enumerate}
To answer the query, we want to know if the number of points of the last type is greater than $0$. We know that the number of points of the fifth category is exactly $|J|$.
Notice that if we now insert $t_{i'}$ to $D$ and get the new count $c'$ of extremal points, we are counting exactly the first four categories of points, where we have replaced $t_{i'+1}$ with $t_{i'}$. Thus, we can test if the number of points of the last category is $0$ simply by testing if $c-c' = |J|$. We can thus perform Step 3 with $O(1)$ updates.

Let us now adapt this to work on a Word RAM machine with words of length $w \geq \log n$. Let $G$ be a $3$-dimensional orthogonal grid with uniform spacing $1/(R^2\sqrt{3})$. Whenever we say we insert some point to $D$, we actually place it at the closest point which lies on a vertex of the grid $G$. Thus, all points lie at a distance of at most $1/R^2$ from their originally intended location. By Lemma \ref{lemma:extremal} this doesn't change which points are extremal at any stage of the procedure. As all points have coordinates with absolute value bounded by $R$ and lie on the vertices of $G$, we can assume by some appropriate scaling that all coordinates are integers with absolute value bounded by $O(R^3) = O(k^6\cdot m^6) = O(n^6)$. This means that all coordinates we consider each can be described with $O(\log n)$ bits and thus fit in a constant number of machine words. There is a final issue which is that we have originally defined the points in cylindrical coordinates. To convert them to Cartesian coordinates requires the computation of $r\cdot \cos(\theta)$ and $r\cdot \sin(\theta)$ for $0\leq \theta \leq 2\pi$ and $0\leq r \leq R$ up to $O(\log n)$ bits of precision. This can be done in $O(\ppolylog(n))$ time. Precomputing this for every point thus costs an additional  $O(n\cdot \ppolylog(n))$ time, which does not affect our result here.

By applying Theorem \ref{thm:general_3sum} and Theorem \ref{thm:general_omv} we then get the result.
\end{proof}
As is the case for maximal points, for purely incremental data structures with per-update runtime guarantees (rather than amortized), we can adapt the proof of Theorem \ref{thm:extremal_pts_ds} so that in Step 1 we do not insert the points $b_\bullet$ and in Step 2 we insert the points $b_{j}$ for all $j \not\in J$ (instead of deleting those for which $j \in J$). This again leads to the same bounds from the OMv problem and slightly worse lower bounds from the Exact Triangle problem (corresponding to Scenario 1 in Theorem \ref{thm:general_3sum}).

Because we can assume $t_p = O(t_u\cdot n)$ in both cases, we get the following corollaries.
\begin{corollary}
Let $D$ be a fully-dynamic data structure for Counting Extremal Points in $\R^3$ with $t_u$ expected amortized update time. If the Exact Triangle conjecture is true then
$t_u + t_q = \Omega\left(n^{1/5-o(1)}\right)$.
\end{corollary}

\begin{corollary}
Let $D$ be an incremental data structure for Counting Extremal Points in $\R^3$ with $t_u$ expected time per update. If the Exact Triangle conjecture is true then
$t_u + t_q = \Omega\left(n^{1/5-o(1)}\right)$.
If the OMv conjecture is true then
$t_u + t_q = \Omega\left(n^{1/2-o(1)}\right)$.
\end{corollary}

If we assume almost linear $n^{1+o(1)}$ preprocessing time then if the Exact Triangle conjecture holds we have 
\[n^{1+o(1)} + t_u\cdot\left(n^{\frac{2-\gamma}{3-2\gamma}}+n^{\frac{1+\gamma}{3-2\gamma}}\right)  + t_q \cdot n^{\frac{1+\gamma}{3-2\gamma}} = \Omega\left(
n^{\frac{2}{3-2\gamma}-o(1)}\right).\]
In particular by letting $\gamma$ approach $1/2$ from above we get the following.
\begin{corollary}
Let $D$ be a fully-dynamic data structure for Counting Extremal Points in $\R^3$ with $n^{1+o(1)}$ expected preprocessing time and $t_u$ expected amortized update time. If the Exact Triangle conjecture is true then
$t_u +t_q = \Omega\left( n^{1/4-o(1)}\right)$.
\end{corollary}

In the fully-dynamic setting, Chan \cite{Chan2020} gives a data structure for this problem with $O(n^{1+\epsilon})$ preprocessing time and $O(n^{11/12+\epsilon})$ amortized update and query time, for an arbitrary $\epsilon>0$. In the more restricted semi-online setting (which generalizes the incremental case), another paper by the same author \cite{Chan2003} gives a data structure with $O(n^{1+\epsilon})$ preprocessing time and $O(n^{7/8+\epsilon})$ worst-case time per operation.

\section{Dynamic geometric Set Cover with squares}\label{sec:square_set_cover}
In this section we answer a question by Chan et al.\ \cite{Chan2022}, by giving a conditional polynomial lower bound on the time required to approximately maintain (the size of) a dynamic square set cover in the plane under range updates.

\subsection{The unweighted case}
\subparagraph*{Dynamic Square Set Cover:} Maintain a set $S$ of $n$ points and axis-aligned squares in the plane to support queries asking for the size of the smallest subset of squares which covers all points.

Even the static version of this problem with unit squares is $\mathrm{NP}$-complete \cite{Fowler1981}, thus the focus on approximations. Chan et al.\ \cite{Chan2022} recently gave a $O(1)$-approximate solution in the fully dynamic case where both squares and points may be inserted or deleted. This (Monte Carlo randomized) solution achieves $O(n^{1/2+\epsilon})$ amortized update and query time. The authors ask if there is a conditional polynomial lower bound for this problem. We show the following.
\begin{theorem}\label{thm:square_set_cover}
Let $0\leq \alpha< 1$ be an efficiently computable\footnote{We say that a number $\alpha$ is efficiently computable if there is an algorithm which, for any $k\geq 0$, can output the first $k$ bits of $\alpha$ in $O(\ppoly(k))$ time.} constant. If there is a fully-dynamic data structure for $O(n^\alpha)$-approximate Dynamic Square Set Cover with $t_u$ expected amortized update time and $t_q$ expected query time, then the Multiphase problem can be solved with $n=O\left(k^{1/(1-\alpha)}\cdot m^{2/(1-\alpha)^2}\right)$, $t_{uq} = t_u+t_q$ and $u_J = O(m)$. 
\end{theorem}

Together with Theorem \ref{thm:general_omv} this implies the following.
\begin{corollary}
Let $0\leq \alpha < 1$ be an efficiently computable constant.\footnote{In this result and in Theorem \ref{thm:weigted_square_set_cover}, we could get rid of the assumption that $\alpha$ is efficiently computable with some more work. This does however not seem worth the trouble, as we can always replace $\alpha$ with a slightly larger but arbitrarily close efficiently computable number and get an almost identical bound.} Let $D$ be a fully dynamic data structure for $O(n^\alpha)$-approximate Dynamic Square Set Cover with $t_u$ expected amortized update time and $t_q$ expected query time. If the OMv-conjecture is true then for any $0<\gamma<1$
\[t_u\cdot n^\gamma + (t_u+t_q) \cdot n^{1-\alpha-\frac{2\gamma}{1-\alpha}} = \Omega\left(n^{1-\alpha-\gamma\frac{1+\alpha}{1-\alpha}-o(1)}\right).\]
In particular, for $\alpha = 0$ (i.e. for a constant approximation factor) and $\gamma = 1/3$ we have
$t_u + t_q = \Omega\left(n^{1/3-o(1)}\right)$.
\end{corollary}
Note that for any $0\leq \alpha < 1$ we can get a polynomial lower bound (whose exponent depends on $\alpha$).

\begin{figure}
    \centering
    \includegraphics[width=\textwidth]{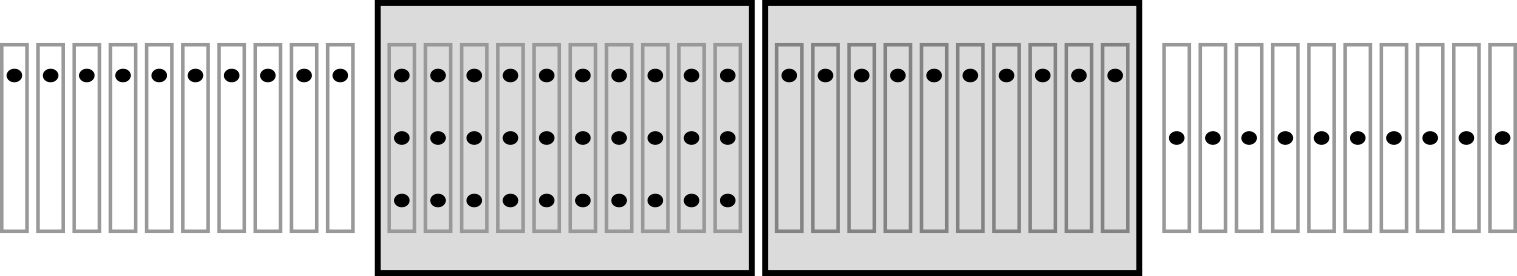}
    \caption{Illustration (not to scale) of a set of points and rectangles obtained after Step 2 in the proof of Theorem \ref{thm:square_set_cover}, for a data structure with a constant approximation ratio of $\beta = 3/2$. The illustrated instance has $m=4$, $k=3$, $\mathcal{F} = \{\{1,2,3\}, \{2,4\}, \{2\} \}$, $J = \{1,4\}$ and $c=10$.}
    \label{fig:square_set_cover}
\end{figure}

\begin{proof}[Proof of Theorem \ref{thm:square_set_cover}]
We first prove the result in the case of arbitrary axis-aligned rectangles and then show how to adapt it to use only squares.
Let $\mathcal{F}=\{F_1, \ldots, F_k\}$ be a family of $k$ subsets of $\{1,2,\ldots m\}$. Suppose $D$ achieves an approximation ratio of at most $\beta\cdot n^\alpha$ for some constant integer $\beta > 1$ (we can assume this without loss of generality). Let $c$ be an integer value which we specify later.

We perform Step 1 by initializing $D$ with the following points and rectangles: 
\begin{itemize}
    \item for each $1\leq i\leq k$ and $1\leq j \leq m$ for which $j \in F_i$, we put a point $p_{i,j}^a$ at coordinates $(c\cdot j+a,i)$ for each $0 \leq a < c$;
    \item for each $1\leq j \leq m$ and each $0 \leq a < c$, we put a thin vertical rectangle which covers exactly the points of the form $p_{\bullet,j}^a$.
\end{itemize}

The total number of points and rectangles at this point is $n_1 = s_\mathcal{F}+m\cdot c \leq m\cdot (k+1) \cdot c$.
Set $c$ to be the smallest integer such that $c > \beta (n_1+m+2)^\alpha \cdot (m+2)$ (note that the value of $n_1$ depends on $c$, but $c$ is nonetheless well defined and can be computed in $O(\ppolylog(m\cdot k))$ time). We then have $n_1 = O\left(k^{1/(1-\alpha)}\cdot m^{2/(1-\alpha)^2}\right)$.

To perform Step 2 when given $J \subset \{1,2,\ldots m\}$, we insert a rectangle covering all points of the form $p_{\bullet,j}^\bullet$ for each $j \not\in J$. This requires $O(m)$ updates on $D$. See Figure \ref{fig:square_set_cover} for an illustration of the first two steps.

In Step 3, when given an index $1\leq i'\leq k$, we insert two rectangles: the first covers all points $p_{i,\bullet}^\bullet$ with $i<i'$ and the second covers all points $p_{i,\bullet}^\bullet$ with $i>i'$. The total number of points and rectangles at this point is $n \leq n_1+m+2$. Now, if $J\cap F_{i'} = \emptyset$ then all points can be covered with at most $m+2$ rectangles: the (at most) $m$ rectangles inserted in step 2 together with the two rectangles inserted in Step 3. On the other hand, if there is some $j \in J\cap F_{i'}$, then the points of the form $p_{i,j}^\bullet$ can only be covered by choosing $c$ thin rectangles created in Step 1. Thus, any approximation of the set cover with a ratio better than $\frac{c}{m+2}$ suffices to distinguish between the two cases. Moreover, we have $\frac{c}{m+2} > \beta\cdot (n_1+m+2)^\alpha\frac{m+2}{m+2} \geq \beta\cdot n^\alpha$. We can then answer a Step 3 query by asking the data structure for a $\beta\cdot n^\alpha$ approximation of the size of the minimum set cover.

To see how to adapt this reduction using only squares, notice that we can increase the height of any rectangle without affecting the results. Thus, we can stretch the whole configuration of points and rectangles horizontally until the thinnest vertical rectangles become squares, and adjust the heights of the other rectangles to make them squares as well.
\end{proof}

Notice that this proof only requires ranges to be inserted or deleted and still works if the set of points is static. It also applies to incremental data structures with per-operation runtime guarantees. We can get the same bound when replacing the squares with orthogonal slabs of the plane (``rectangles'' which are unbounded in two opposite directions).

Contrast this lower bound with the case of \emph{unit} axis-aligned squares, for which Chan et al. \cite{Chan2022} give a data structure for $O(1)$-approximation achieving $2^{O(\sqrt{\log n})}$ amortized update and query time.

\subsection{The weighted case}
In the case where we associate to each square a polynomially bounded integer weight and ask for an approximate weight for the minimum weight set cover, we can modify the reduction as follows.
\begin{itemize}
    \item Replace the points $p_{i,j}^a$ for $0 \leq a < c$ with a single point $p_{i,j}$.
    \item In Step 2, for each $j\not\in J$, replace the corresponding $c$ thin vertical rectangles with a single rectangle covering all points of the form $p_{\bullet,j}$ and weight $c$.
\end{itemize}
In the case of a fully-dynamic data structure, we can also pre-insert rectangles in Step 1 and delete some of them in Step 2, as we have done for previous reductions. Then by a similar argument, we get the following.
\begin{theorem}
Let $0\leq \alpha < 1$ be an efficiently computable constant. If there is an incremental data structure for $O(n^\alpha)$-approximate Weighted Dynamic Square Set Cover with $t_u$ expected time per update and $t_q$ expected time per query, then Scenario 3 in Theorem \ref{thm:general_3sum} applies (with $t_{uq} = t_u + t_q$).

If the data structure is fully-dynamic then Scenario 4 applies,  even for amortized runtime guarantees.
\end{theorem}\label{thm:weigted_square_set_cover}
In the fully dynamic case for example, Theorem \ref{thm:general_3sum} and Theorem \ref{thm:general_omv} thus imply the following lower bounds.
\begin{theorem}
Let $0\leq \alpha < 1$ be an efficiently computable constant. Let $D$ be a fully-dynamic data structure for $O(n^\alpha)$-approximate Weighted Dynamic Square Set Cover with $t_u$ expected amortized update time and $t_q$ expected amortized query time. If the Exact Triangle conjecture holds, then for any $0 < \gamma < 1$
\[t_u\cdot(n + n^\frac{1+\gamma}{2-\gamma})+t_q\cdot n^\frac{1+\gamma}{2-\gamma} = \Omega\left( n^{\frac{2}{2-\gamma}-o(1)}\right).\]
If the OMv conjecture holds, then for any $0 < \gamma < 1$
\[t_u\cdot (n^{1 - \gamma} + n^{\gamma}) + t_q \cdot n^{1-\gamma} = \Omega\left(n^{1-o(1)}\right).\]
In particular, for $\gamma = 1/2$ we have
$t_u + t_q = \Omega\left( n^{1/2 - o(1)}\right)$. 
\end{theorem}

Contrast this lower bound with the weighted case of \emph{unit} axis-aligned squares, for which Chan et al. \cite{Chan2022} give a data structure for $O(1)$-approximation achieving $O(n^{\epsilon})$ amortized update and query time for an arbitrarily small constant $\epsilon > 0$.

\section{An unconditional lower bound for incremental Hypervolume Indicator Problem in \texorpdfstring{$\R^3$}{R3}} \label{subsec:hypervolume_indicator}

Here we give an unconditional lower bound for the incremental variant of Klee's Measure Problem in $\R^3$ where all boxes are in the positive orthant and have one vertex lying on the origin. This special case is known as the Hypervolume Indicator Problem (in $\R^3$). It has applications to the evaluation of multiobjective optimization algorithms and has been the object of many papers (see the survey on the topic by Guerreiro et al.\ \cite{GuerreiroFP2021}). We show the lower bound by a reduction from the Dynamic Matrix–Vector Multiplication Problem \cite{Frandsen2001}. In this problem, we want to maintain a dynamic matrix $M$ of size $N\times N$ and a vector $v$ of size $N$ with updates consisting of changing an entry of $M$ or $v$ and queries consisting of computing a given entry of the product $Mv$. Frandsen et al.\ showed a $\Omega(N)$ lower bound on the worst-case update time per operation in various general general models of computation such as history dependent algebraic computation trees and the Real RAM. In particular, they showed that in the case of non-negative integer inputs bounded by $2^w - 1$, the lower bound holds on the Word RAM model with words of size $w\geq \log n$ (and holds even if we are given arbitrary preprocessing time for the initial state of $M$ and $v$).\footnote{Here, contrary to everywhere else in this paper, we assume only $w\geq\log n$ and not $w=O(\log n)$, which is essential to apply the result of Frandsen et al.} 
A closer inspection of their proof reveals that they show the following stronger statement.
\begin{theorem}[{\cite{Frandsen2001}}] \label{thm:matrix_vector_mult}
Consider the following problem in the Word RAM model with $w\geq \log n$:
\begin{itemize}
    \item First, we are given an $N\times N$ matrix and $O(\ppoly(N))$ time to read and preprocess its entries.
    \item After this is done, we are given a vector $v$ of size $N$ and return the product $Mv$ after $O(t)$ time.
\end{itemize}
Then for any procedure solving this problem we have $t = \Omega(N^2)$ in the worst case.
\end{theorem}

Our reduction will be from this form of the problem, and is inspired by a reduction by Chan \cite{Chan2010}, from the original form of dynamic Matrix–Vector Multiplication to dynamic Klee's Measure Problem with arbitrary axis-aligned rectangles in the plane.
We show the following.
\begin{theorem}\label{thm:hypervolume_indicator}
Let $D$ be an incremental data structure, in the Word RAM model with $w\geq \log n$, which maintains a set of axis-aligned boxes in $\R^3$ with non-negative integer vertex-coordinates which all have one vertex at the origin, together with the volume of their union. If the preprocessing time for $D$ is at most polynomial in $n$, then the amortized update time for $D$ is $\Omega(\sqrt{n})$.
\end{theorem}

Bringmann \cite{Bringmann2013} showed that the Hypervolume Indicator Problem reduces efficiently to Klee's Measure Problem for unit (hyper)cubes in the same number of dimensions (and this reduction also works in the dynamic setting we consider here). This former problem is thus the easiest among the most commonly considered special cases of Klee's Measure Problem and we immediately get the same lower bound for these other variants.

\begin{proof}
Consider a matrix $M$ of size $N\times N$ with non-negative integer entries bounded by $W = 2^w-1$. We start by initializing the data structure in a preprocessing phase with $\Theta(N^2)$ boxes. Because we always have one vertex at the origin, we only specify the coordinate of the opposing vertex in what follows.
\begin{itemize}
    \item For all $1 \leq i,j \leq N$, put a box with vertex opposing the origin at coordinates $(jW, iW, N(N+1)- i - jN)$.
    \item For all $1 \leq i,j \leq N$, put a box corresponding to $M_{i,j}$ with vertex opposing the origin at coordinates $(jW, (i-1)W+M_{i,j}, N(N+1)-i - jN+1)$.
\end{itemize}

We also compute the sums $\sum_{j=1}^N M_{k,j}$ for all $1\leq k \leq N$ in $O(N^2)$ time.

When given a vector $v$ of size $N$, we insert a box corresponding to $v_j$ with vertex opposing the origin at coordinates $((j-1)W+v_j, N\cdot W, N(N+1) - j\cdot N)$, for all $1 \leq j \leq N$. This costs $N$ updates. We also compute $\sum_{j=1}^N v_j$ in $O(N)$ time.

The volume $c_0$ of the union of the boxes at this point is 
\begin{align*}
    \sum_{j=1}^N\sum_{i=1}^N (N(N+1)- i - jN)\cdot W^2 + \sum_{j=1}^N\sum_{i=1}^N M_{i,j}\cdot W 
    + \sum_{j=1}^N v_j \sum_{i=1}^N i\cdot W - \sum_{j=1}^N v_j\sum_{i=1}^N M_{i,j}.
\end{align*}

Now, to compute the first entry of $Mv$, insert a box whose vertex opposing the origin lies at coordinate $(NW,W, N(N+1))$ and compute the new total volume $c_1$. This volume is exactly 
\begin{align*}
    c_1 = N^2(N+1)W^2 
    &+ \sum_{j=1}^N\sum_{i=2}^N (N(N+1)- i - jN)\cdot W^2 
    + \sum_{j=1}^N\sum_{i=2}^N M_{i,j}\cdot W \\
    &+ \sum_{j=1}^N v_j \sum_{i=2}^N i\cdot W
    - \sum_{j=1}^N v_j\sum_{i=2}^N M_{i,j}.
\end{align*}

The difference is thus exactly
\begin{align*}
    c_1-c_0 =& 
    N^2(N+1)W^2 
    - \sum_{j=1}^N (N(N+1)- 1 - jN)\cdot W^2 
    - \sum_{j=1}^N M_{1,j}\cdot W\\
    &- \sum_{j=1}^N v_j \cdot W 
    + \sum_{j=1}^N v_j M_{1,j} \\
    =&
    \left(1 +\frac{N(N+1)}{2}\right)NW^2 
    - W\sum_{j=1}^N M_{1,j} 
    - W\sum_{j=1}^N v_j
    + \sum_{j=1}^N v_j M_{1,j}.
\end{align*}

The sums $\sum_{j=1}^N M_{1,j}$ and $\sum_{j=1}^N v_j$ are known. Thus, knowing $c_1$ and $c_0$ we can compute $\sum_{j=1}^N v_j M_{1,j}$ in constant additional time.

To compute the second entry of $Mv$, insert a box whose vertex opposing the origin lies at coordinate $(NW,2W, N(N+1))$, compute the new total volume $c_2$, and proceed similarly. In general, to compute the $k$'th entry of $Mv$, insert a box whose vertex opposing the origin lies at coordinate $(NW,kW, N(N+1))$, compute the new total volume $c_k$.
The difference with the previous volume is then 
\begin{align*}
    c_k-c_{k-1} =
    \left(k +\frac{N(N+1)}{2}\right)NW^2 
    - W\sum_{j=1}^N M_{k,j} 
    - kW\sum_{j=1}^N v_j
    + \sum_{j=1}^N v_j M_{k,j}.
\end{align*}
Knowing this difference we can compute $\sum_{j=1}^N v_j M_{k,j}$ in constant additional time.

Thus, computing $v$ can be done in $O(t_u\cdot N + N)$ time, where $t_u$ is the amortized time to perform an update. Theorem \ref{thm:matrix_vector_mult} then implies $t_u = \Omega(N)$. Written in terms of $n$, the total number of boxes stored, this is $t_u = \Omega(\sqrt{n})$.
\end{proof}
Note that the bound remains valid in the case of worst-case runtime for data structures with no preprocessing operation, by replacing the preprocessing step with insertions. This does not work for amortized runtime however, as the updates performed to compute $Mv$ could be high cost updates amortized against low cost updates performed in the ``preprocessing'' phase.

In a previous version of this paper, we proved the bound only for worst-case time in the fully-dynamic setting. After that version appeared on arXiv, Jin and Xu \cite{Jin2022} independently gave lower bounds for the dynamic Klee's Measure Problem with unit (hyper)cubes in odd dimension, conditioned on a generalization of the OMv conjecture. In dimension $3$, they obtain an $\Omega(n^{1/2-o(1)})$ lower bound on the amortized update time in the semi-online setting. Contrast this with our current bound of $\Omega(n^{1/2})$, which is unconditional and applies already to the easier Hypervolume Indicator problem in the incremental setting. It would be interesting to know if variants of Frandsen et al.'s result could be used to make the lower bounds in higher dimension of Jin and Xu unconditional.

On the positive side, Chan \cite{Chan2003} gives a data structure for this problem in the incremental setting (or the more general semi-online setting) with $O(n^{1+\epsilon})$ preprocessing time and $O(\sqrt{n}\ppolylog n)$ worst-case time per update (even for the more general case of axis-aligned cubes). In the fully-dynamic setting, another paper by the same author \cite{Chan2020} gives a data structure with $O(n\cdot\ppolylog n)$ preprocessing time and $O(n^{2/3}\ppolylog n)$ amortized update time.

\section{Further results}\label{sec:further_results}

The same approach as for counting maximal or extremal points can be used for many other dynamic problems. For the sake of conciseness, we give only the Scenarios in Theorem \ref{thm:general_3sum} which fit the data structures, from which the reader can easily find the implied lower bounds.

\subsection{Measure problems for axis-aligned squares and rectangles} \label{subsec:rectangles}

\subparagraph*{Klee's Measure Problem with Squares:} Maintain a dynamic set of $O(n)$ axis-aligned squares $S$ and support queries that return the area of their union, $\cup S$.

\begin{lemma}\label{lemma:kmp}
For any fully-dynamic data structure $D$ for Klee's Measure Problem with Squares, Scenario 4 in Theorem \ref{thm:general_3sum} applies (with $t_{uq} = t_u + t_q$).
\end{lemma}

\begin{figure}
    \centering
    \includegraphics[scale=0.4]{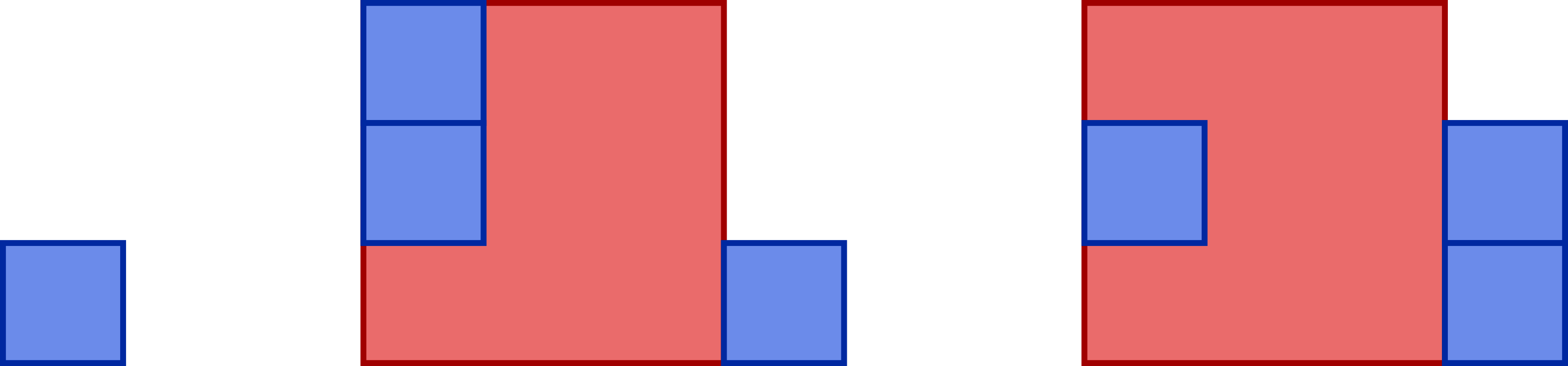}
    \caption{Illustration of a set of rectangles obtained after Step 2 in the proof of Lemma \ref{lemma:kmp}. The illustrated instance has $m=5$, $k=3$, $\mathcal{F} = \{\{1,3,5\}, \{2,4,5\}, \{2\} \}$ and $J = \{1,3,5\}$. The large red squares correspond to $b_{\bullet}$.}
    \label{fig:kmp}
\end{figure}

\begin{proof}
Let $\mathcal{F}=\{F_1, \ldots, F_k\}$ be a family of $k$ subsets of $\{1,2,\ldots m\}$. Suppose without loss of generality that no set in $\mathcal{F}$ is empty and every element in $\{1,2,\ldots m\}$ appears in at least one set.

We perform Step 1 by initializing $D$ with all the following squares.
\begin{itemize}
    \item For each $1\leq i\leq k$ and $1\leq j \leq m$ for which $j \in F_i$, a unit square whose lower-left corner has coordinates $(kj,i)$.
    \item For each $1\leq j \leq m$, a square $b_j$ of side-length $k$ whose lower-left corner has coordinates $(kj,1)$.
\end{itemize}

The total number of squares is $O(s_\mathcal{F})$. It is easy to see that at this point, the area of the union of all squares is $k^2m$.

To perform Step 2 when given $J \subset \{1,2,\ldots m\}$, we delete the squares $b_{j}$ for all $j\in J$. This requires $O(|J|)$ updates on $D$. See Figure \ref{fig:kmp} for an illustration.

In Step 3, when given an index $1\leq i'\leq k$, we insert two squares of side-length $km$ whose lower-left corners lie at coordinates $(0,i'+1)$ and $(0,i'-km)$ respectively. Now the area of the union of all squares in $S$ is strictly larger than $2(km)^2 + (m-|J|)k$ if and only if there is some square which is not contained in the union of these two new squares and the squares $\{b_j|j\not\in J\}$. This is the case if and only if there is some $j$ such that $j\in J$ and $j\in F_{i'}$ (i.e.\  $J\cap F_{i'} \neq \emptyset$). Thus, we can answer such a Step 3 query after only two more updates to the data structure.
\end{proof}

A straight-forward adaptation of this proof yields the following.
\begin{lemma}
For any incremental data structure $D$ for Klee's Measure Problem with Squares with per-operation runtime guarantees, Scenario 3 in Theorem \ref{thm:general_3sum} applies (with $t_{uq} = t_u + t_q$).
\end{lemma}

As previously stated, Jin and Xu \cite{Jin2022} recently and independently gave lower bounds for the dynamic Klee's Measure Problem with unit (hyper)cubes in odd dimension, conditioned on a generalization of the OMv conjecture. 

Overmars and Yap \cite{Overmars1988} show how the more general problem for rectangles can be solved in the fully-dynamic setting with amortized update time $O(\sqrt{n}\log n)$, assuming that the set of all rectangles which will be inserted or deleted is of size $O(n)$ and that the set of their vertices is known \emph{a priori}. Chan \cite{Chan2010} improved the amortized update time to $\sqrt{n}2^{O(\log^*n)}$. Yıldız et al.\ \cite{Yildiz2011-prob} showed how to achieve $O(\sqrt{n}\log n)$ amortized update time without these additional assumptions (in fact, they show this for a more general problem where each rectangle has an associated probability of being present and we want to compute the expected volume of their union).
The bounds we obtain in this setting from the OMv conjecture are thus tight up to a factor of $n^{o(1)}$.

For the fully-dynamic version of Klee's measure problem and arbitrary axis-aligned rectangles, Chan \cite{Chan2010} also gives an unconditional $\Omega(\sqrt{n})$ lower bound on the worst-case update time by a reduction from the Dynamic Matrix–Vector Multiplication Problem \cite{Frandsen2001}. While this bound is certainly more powerful in the sense that it is unconditional, it does not apply when restricted to squares, nor for amortized runtime or partially-dynamic data structures. We could however adapt this lower bound for amortized runtime in the incremental setting by an argument similar to the one used earlier in the paper for the Hypervolume Indicator Problem.

By a similar proof we also get a lower bound for a discrete version of the problem.
\subparagraph*{Discrete version of Klee's Measure Problem with Squares:}
Maintain a dynamic set of $O(n)$ points $P$ and $O(n)$ axis-aligned squares $S$ in the plane and support queries that return the number of points in $P$ covered by squares in $S$.

\begin{lemma}
For any fully-dynamic data structure $D$ for the discrete version of Klee's Measure Problem with Squares, Scenario 4 in Theorem \ref{thm:general_3sum} applies (with $t_{uq} = t_u + t_q$).
\end{lemma}
\begin{proof}
Let $\mathcal{F}=\{F_1, \ldots, F_k\}$ be a family of $k$ subsets of $\{1,2,\ldots m\}$. Suppose without loss of generality that no set in $\mathcal{F}$ is empty and every element in $\{1,2,\ldots m\}$ appears in at least one set.

We perform Step 1 by initializing $D$ with the following points.
\begin{itemize}
    \item For each $1\leq i\leq k$ and $1\leq j \leq m$ for which $j \in F_i$, we put a point $p_{i,j}$ at coordinates $((k+2)j+1,i+1)$.
    \item For each $1\leq j \leq m$, a square $b_j$ of side-length $k+2$ whose lower-left corner has coordinates $((k+2)j,1)$.
\end{itemize}

The total number of points and squares is $O(s_\mathcal{F})$. 

To perform Step 2 when given $J \subset \{1,2,\ldots m\}$, we delete the square $b_j$ for all $j \not\in J$. This requires $O(|J|)$ updates on $D$. Now the uncovered points are exactly the points $p_{i,j}$ such that $j\in J$.

In Step 3, when given an index $1\leq i'\leq k$, we insert two squares of side-length $(k+2)m$ whose lower-left corners lie at coordinates $(0,i'+ 1/2)$ and $(0,i'-(k+2)m-1/2)$ respectively.
Now there is an uncovered point if and only if there is some point which is not covered by these two new squares or the squares $\{b_j|j\not\in J\}$. This is the case if and only if there is some $j$ such that $j\in J$ and $j\in F_{i'}$ (i.e.\  $J\cap F_{i'} \neq \emptyset$). We can easily test this by counting the number of covered points and comparing it to the total number of points. Thus, we can answer such a Step 3 query after only two more updates to the data structure.
\end{proof}

A straight-forward adaptation of this proof yields the following.
\begin{lemma}
For any incremental data structure $D$ for the Discrete version of Klee's Measure Problem with Squares with per-operation runtime guarantees, Scenario 3 in Theorem \ref{thm:general_3sum} applies (with $t_{uq} = t_u + t_q$).
\end{lemma}

Yıldız et al.\ \cite{Yildiz2011} give a data structure for this problem in the fully-dynamic setting with $O(\sqrt{n})$ worst-case time per update, even when considering arbitrary axis-aligned rectangles instead of squares. This matches the lower bound obtained from the OMv conjecture up to a $n^{o(1)}$ term.

\subparagraph*{Depth Problem with Squares}
Maintain a dynamic set of $O(n)$ axis-aligned squares $S$ in the plane and support queries that return their depth.

\medskip
Recall that given a set $S$ of subsets of $\R^d$, the depth of a point $p\in \R^d$ with respect to $S$ is the number of sets of $S$ which contain $p$. The depth of $S$ is the maximum depth with respect to $S$ over all points in $\R^d$.

\begin{lemma}
For any fully-dynamic (or incremental with per-operation runtime guarantees) data structure $D$ for the Depth Problem with Squares, Scenario 4 in Theorem \ref{thm:general_3sum} applies (with $t_{uq} = t_u + t_q$).
\end{lemma}
\begin{proof}
Let $\mathcal{F}=\{F_1, \ldots, F_k\}$ be a family of $k$ subsets of $\{1,2,\ldots m\}$. Assume without loss of generality that no set in $\mathcal{F}$ is empty and every element in $\{1,2,\ldots m\}$ appears in at least one set.

We perform Step 1 by initializing $D$ with all the following squares.
\begin{itemize}
    \item For each $1\leq i\leq k$ and $1\leq j \leq m$ for which $j \in F_i$, a unit square $s_{i,j}$ whose lower-left corner has coordinates $(kj,i)$.
\end{itemize}

The total number of squares is $O(s_\mathcal{F})$. It is easy to see that at this point, the area of the union of all squares is $k^2m$.

To perform Step 2 when given $J \subset \{1,2,\ldots m\}$, we put a square $b_j$ of side-length $k$ whose lower-left corner has coordinates $(kj,1)$ for all $j\in J$. This requires $O(|J|)$ updates on $D$.

In Step 3, when given an index $1\leq i'\leq k$, we insert two squares of side-length $km$ whose lower-left corners lie at coordinates $(k,i')$ and $(k,i'+1-km)$ respectively. Now the depth of $S$ is $4$ if and only if there is some $j$ such that $j\in J$ and $j\in F_{i'}$ (i.e.\  $J\cap F_{i'} \neq \emptyset$). Thus, we can answer an intersection query after only two more updates to the data structure in Step 3.
\end{proof}

For this specific problem, we can also get a lower bound from the OMv conjecture in the case of incremental data structure with amortized runtime bounds, by simulating the deletion of squares by inserting more squares.
\begin{lemma}
Let $D$ be an incremental data structure for the Depth Problem with Squares with $t_u$ expected amortized update time.
If the OMv problem on matrices of size $N\times N$ requires expected time $\Omega(N^{3-\delta})$, then
\[t_u + t_{q} = \Omega\left(n^{(1-\delta)/2}\right).\]
In particular, if the OMv conjecture holds,
\[t_u + t_q = \Omega\left( n^{1/2-o(1)}\right).\]
\end{lemma}
\begin{proof}
We will use this data structure to solve the OMv problem. Let $M$ be a boolean matrix of size $N\times N$, and let $v^1, v^2,\ldots, v^N$ be boolean vectors of size $N$. 
We start by inserting in $D$ all the following squares.
\begin{itemize}
    \item For each $1\leq i\leq N$ and $1\leq j \leq M$ for which $M_{i,j} = 1$, we put a unit square $s_{i,j}$ whose lower-left corner has coordinates $(Nj,i)$.
\end{itemize}
This takes $O(N^2)$ updates on $D$.

Then we are given $v^1$ and compute $Mv^1$ as follows. Insert in $D$ squares of side-length $N$ whose lower-left corner have coordinates $(Nj,1)$ for all $j$ such that $v^1_j = 1$. At this point the depth of any point in some square $s_{i,j}$ is $2$ if  $v^1_j = 1$ and $1$ otherwise. Any point which does not lie in a square $s_{\bullet,\bullet}$ has depth at most $1$. To compute $(Mv^1)_1$, insert in $D$ two squares with side-length $N^2$ and lower-left corner at $(N,1)$ and $(N,2-N^2)$ respectively. Then $(Mv^1)_1=1$ if and only if the depth of $S$ is $4$ and we can thus compute this entry. Next, we insert in $D$ two squares with side-length $N^2$ and lower-left corner at $(N,2)$ and $(N,1-N^2)$ respectively. At this point, all points which lie in a square $s_{\bullet,j}$ with $v^1_j = 1$ have depth $4$ while all other points have depth at most $3$. 

To compute $(Mv^1)_2$, insert in $D$ two squares with side-length $N^2$ and lower-left corner at $(N,2)$ and $(N,3-N^2)$ respectively. Then $(Mv^1)_2=1$ if and only if the depth of $S$ is $6$ and we can thus compute this entry. Again insert in $D$ two squares with side-length $N^2$ and lower-left corner at $(N,3)$ and $(N,2-N^2)$ respectively. Now, all points which lie in a square $s_{\bullet,j}$ with $v^1_j = 1$ have depth $6$ while all other points have depth at most $5$. 

We continue this way to compute $(Mv^1)_j$ for all $1\leq j\leq N$. Once this is done, insert in $D$ squares of side-length $N$ whose lower-left corner have coordinates $(Nj,1)$ for all $j$ such that $v^1_j = 0$. Now, all points which lie in a square $s_{\bullet,\bullet}$ have depth $2(N+1)$ while all other points have depth at most $2(N+1)-1$. When given $v^2$ we can restart this whole procedure for $v^2$ (adding $2(N+1)-1$ to the depth threshold for every test).

We repeat this with $v^i$ for all $3\leq i \leq N$, thus solving the OMv problem.

All in all, we have performed $\Theta(N^2)$ updates and queries on $D$ (and the total number of squares is also $n = \Theta(N^2)$). Thus, we can solve the OMv problem on matrices of size $N\times N$ in expected time \[O\left((t_q+t_u)\cdot N^2\right).\]

Assuming this is $\Omega(N^{3-\delta})$ and rewriting in terms of $n$ we get the result.
\end{proof}

A slight adaptation of the proof by Overmars and Yap \cite{Overmars1988}  for Klee's Measure Problem again gives an upper bound of $O(\sqrt{n}\log n)$ for the amortized update time in the fully-dynamic setting when the set of rectangle vertices is known \emph{a priori} (Chan \cite{Chan2010} improves this to $O(\sqrt{n/\log n}\log^{3/2}\log n)$). The techniques by Yıldız et al.\ \cite{Yildiz2011-prob} allow the $O(\sqrt{n}\log n)$ result to carry over to the setting without this additional assumption, thus almost matching the lower bound we get from the OMv conjecture.

\subparagraph*{Square Covering with Squares}
Given some fixed square in the plane $C$, maintain a set of $O(n)$ squares $S$ all lying inside $C$ together with the answer to the answer to the question ``is the union of all squares in $S$ equal to $C$?''.

\medskip
This problem is a special case of both Klee's Measure Problem and the Depth Problem (with rectangles).
The reduction to Klee's Measure problem is immediate. To reduce Square Covering to the Depth Problem with Rectangles, replace each square with its complement in the square $C$ (which can be decomposed into at most $4$ rectangles) and test if the depth is $n$.

\begin{lemma}
For any fully-dynamic data structure $D$ for Square Covering with Squares, Scenario 2 in Theorem \ref{thm:general_3sum} applies (with $t_{uq} = t_u + t_q$).
\end{lemma}
\begin{proof}
Let $\mathcal{F}=\{F_1, \ldots, F_k\}$ be a family of $k$ subsets of $\{1,2,\ldots m\}$.

Assume without loss of generality that $C$ is the square of side-length $2R + k$ whose lower-left corner lies at the origin, where $R=m(k+1)$.

We perform Step 1 by initializing $D$ with all the following squares.
\begin{itemize}
    \item Two squares of side-length $R$, one of which has its lower-left corner lies at the origin and the other having lower-left corner at coordinates $(0,R+k)$.
    \item Two squares of side-length $R+k$, one of which has its lower-left corner at coordinates $(R,0)$ and the other having lower-left corner at coordinates $(R,R)$.
    \item For each $1\leq j \leq m$, a square of side-length $k$ whose lower-left corner has coordinates $((k+1)(j-1),R)$.
    \item For each $1\leq i\leq k$ and $1\leq j \leq m$ for which $j \not\in F_i$, a unit square whose lower-left corner has coordinates $((k+1)j-1,R+i-1)$.
    \item For each $1\leq j \leq m$, a square $b_j$ of side-length $k$ whose lower-left corner has coordinates $((k+1)j-1,R)$.
\end{itemize}

The total number of squares is $O(m\cdot k)$. It is easy to see that at this point, the union of all squares in $S$ is indeed equal to $C$.

To perform Step 2 when given $J \subset \{1,2,\ldots m\}$, we delete the square $b_{j}$ for all $j\in J$. This requires $O(|J|)$ updates on $D$. At this point, the only parts of $C$ which are not in the union of the squares in $S$ correspond to the unit squares whose lower left corner is $((k+1)j-1,R+i-1)$ for all $i,j$ such that $j \in J$ and $j\in F_i$.

In Step 3, when given an index $1\leq i'\leq k$, we insert two squares of side-length $R$ whose lower-left corners lie at coordinates $(0,i'-1)$ and $(0,R+i')$ respectively. Now the union of all squares in $S$ is different from $C$ if and only if there is some $j$ such that $j\in J$ and $j\in F_{i'}$ (i.e.\  $J\cap F_{i'} \neq \emptyset$). Thus, we can answer such a query after only two updates to the data structure. 
\end{proof}

A straight-forward adaptation of this proof yields the following.
\begin{lemma}
For any incremental data structure $D$ for Square Coverage by Squares with per-operation runtime guarantees, Scenario 1 in Theorem \ref{thm:general_3sum} applies (with $t_{uq} = t_u + t_q$).
\end{lemma}

For the analogous problem with arbitrary axis-aligned rectangles instead of squares (which we call Square Covering with Rectangles), we can get a stronger condition on the relation between update and query time using Theorem \ref{thm:general_oumv}.

\begin{lemma}
Let $D$ be an incremental data structure for Square Covering with Rectangles with $t_p$ expected preprocessing time, $t_u$ expected time per update and $t_q$ expected time per query. If the OMv conjecture is true and $t_p$ is at most polynomial then
\[t_u\cdot \sqrt{n} + t_q \geq n^{1-o(1)}.\]
The same holds for any fully-dynamic data structure with amortized runtime guarantees.
\end{lemma}

\begin{proof}
We show that such a data structure fits the conditions of Theorem \ref{thm:general_oumv}. Let $\mathcal{F}=\{F_1, \ldots, F_N\}$ be a family of $N$ subsets of $\{1,2,\ldots N\}$.

Assume without loss of generality that $C$ is the square of side-length $N$ whose lower-left corner lies at coordinate $(1,1)$.

We perform Step 1 by initializing $D$ with the following rectangles.
\begin{itemize}
    \item For each $1\leq i\leq N$ and $1\leq j \leq N$ for which $j \not\in F_i$, we put a unit square $s_{i,j}$ whose lower-left corner lies at $(j,i)$.
\end{itemize}

The total number of squares is $n = O(N^2)$. 

To perform Step 2 when given subsets  $I,J \subset \{1,2,\ldots N\}$, we insert for each $i \not\in I$ a long horizontal rectangle of height $1$ and width $N$ whose lower-left corner lies at $(1,i)$. For each $j \not\in J$ we insert a long vertical rectangle of height $N$ and width $1$ whose lower-left corner lies at $(j,1)$. This requires $O(N)$ updates on $D$.

Now the union of all rectangles is different from $C$ if and only if there is $i\in I$ and $j\in J$ such that $j\in F_{i}$. Thus, we can answer such a Step 3 query after a single query to $D$.

By applying Theorem \ref{thm:general_oumv} we get the result.
\end{proof}

As this problem is a special case of Klee's Measure Problem and the Depth Problem with Rectangles, we get the same lower bound for those. This also implies that just like these problems, Square Covering with Rectangles can be solved in the fully-dynamic setting in amortized time $O(\sqrt{n}\log n)$ \cite{Yildiz2011-prob}.

\subsection{Largest empty disk}\label{subsec:disks}

\subparagraph*{Largest Empty Disk in Query Region:} Maintain a set of points in $\R^2$ to support queries for the radius of the largest empty disk whose center lies in a query axis-aligned rectangle.

\begin{lemma}\label{lemma:empty_disk}
For any decremental (or fully-dynamic) data structure $D$ for Largest Empty Disk in Query Region, Scenario 2 in Theorem \ref{thm:general_3sum} applies (with $t_{uq} = t_q$).
\end{lemma}

\begin{figure}
    \centering
    \includegraphics[scale=0.5]{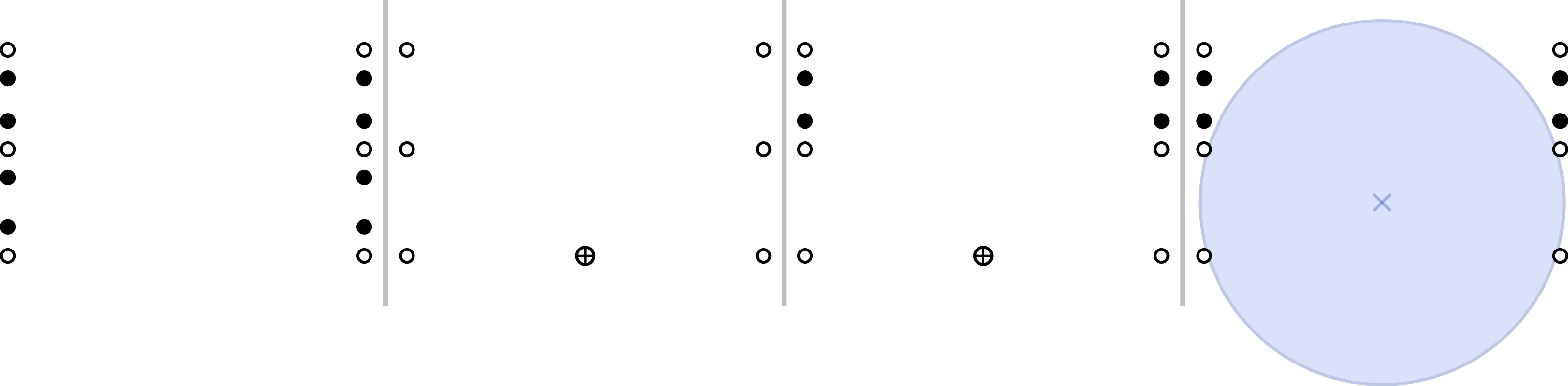}
    \caption{Illustration (not to scale)of a set of points obtained after Step 2 in the proof of Lemma \ref{lemma:empty_disk}. The illustrated instance has $m=4$, $k=2$, $\mathcal{F} = \{\{2,3,4\}, \{2\} \}$ and $J = \{1,4\}$. The points $b_{\bullet}$ are represented by a cross in a circle. The largest empty disk is also represented.}
    \label{fig:empty_disk}
\end{figure}

\begin{proof}
Let $\mathcal{F}=\{F_1, \ldots, F_k\}$ be a family of $k$ subsets of $\{1,2,\ldots m\}$.

We perform Step 1 by initializing $D$ with all the following points.
\begin{itemize}
    \item For all $1 \leq i \leq k+1$ and $1\leq j \leq m$, put points at coordinates $((10k+1)j,4i)$ and $((10k+1)j+10k,4i)$.
    \item For all $1 \leq i \leq k$ and $1\leq j \leq m$ such that $j \not\in F_i$ put points at coordinates $((10k+1)j,4i+1)$,$((10k+1)j,4i+3)$, $((10k+1)j+10k,4i+1)$ and $((10k+1)j+10k,4i+3)$.
    \item For all $1\leq j \leq m$, we put a point $b_j$ at coordinates $((10k+1)j+5k, 4)$.
\end{itemize}

The total number of points is $O(m\cdot k)$. Note that before we insert the points $b_\bullet$, all largest empty disks with center laying inside the convex hulls of the points have their center at $((10k+1)j+5k,4i+2)$ for some $1 \leq i \leq k$ and $1\leq j \leq m$. Moreover, if $j \in F_i$, then such a disk has squared radius $25k^2 + 4$ whereas it has squared radius  $25k^2 + 1$ if $j \not\in F_i$. The point $b_j$ ensures that the squared radius of the largest empty disk with center $((10k+1)j+5k,4i+2)$ is strictly smaller than $25k^2 + 4$ (as it is at most $(4k-1)^2$).

To perform Step 2 when given $J \subset \{1,2,\ldots m\}$, we delete the points $b_{j}$ for all $j\in J$. This requires $O(|J|)$ updates on $D$. See Figure \ref{fig:empty_disk} for an illustration.

In Step 3, when given an index $1\leq i'\leq k$, we query the data structure to get the radius $r$ of the largest empty disk whose center lies in the region $[10k+1; (10k+11)k] \times [4i'; 4(i'+1)]$. To test whether there is some $j$ such that $j\in J$ and $j \in F_{i'}$, it suffices to test if $r^2 = 25k^2 + 4$. Thus we can do Step 3 with a single query to our data structure.
\end{proof}
A straight-forward adaptation of this proof yields the following.
\begin{lemma}
For any incremental data structure $D$ for Largest Empty Disk in Query Region with per-operation runtime guarantees, Scenario 1 in Theorem \ref{thm:general_3sum} applies (with $t_{uq} = t_q$).
\end{lemma}

In the fully-dynamic setting, Chan \cite{Chan2020} gives a data structure for this problem with $O(n^{1+\epsilon})$ preprocessing time and $O(n^{11/12+\epsilon})$ amortized time per operation, for an arbitrary $\epsilon>0$ (here we mean amortized over all operations, and not amortizing queries only over previous queries and updates only over previous updates). In the more restricted semi-online setting (which generalizes the incremental case), another paper by the same author \cite{Chan2003} gives a data structure with $O(n^{7/8+\epsilon})$ worst-case time per update and query. These data structures even support querying arbitrary triangular ranges in the plane (rather than only axis-aligned rectangles).

\subparagraph*{Largest Empty Disk in a Set of Disks:}
Given a fixed axis-aligned rectangle $B$, maintain a set $S$ of disjoint disks in $\R^2$ together with the radius of the largest disk whose center lies in $R$ and which does not intersect any disk in $S$.

\begin{lemma}
For any fully-dynamic data structure $D$ for Largest Empty Disk in a Set of Disks, Scenario 2 in Theorem \ref{thm:general_3sum} applies (with $t_{uq} = t_u + t_q$).
\end{lemma}
\begin{proof}[Proof sketch]
The proof is similar to the previous one (where we use small enough disks instead of points and we let $B$ be the smallest rectangle bounding these points) up to Step 3. Here, instead of using a query to constrain the center of any empty disk of interest to have a $y$-coordinate in $[4i'; 4(i'+1)]$ we use two big disjoint disks $C_1$ and $C_2$, one above and one below all other disks. We place them in such a way that any potential empty disk of squared radius $25k^2 + 4$ whose center lies in $R$ must necessarily have a $y$-coordinate in $[4i'; 4(i'+1)]$ in order to be disjoint from $C_1$ and $C_2$. We then have that the squared radius of a largest empty disk whose center lies inside $B$ is $r^2 = 25k^2 + 4$ if and only if $J\cap F_{i'} \neq \emptyset$.
\end{proof}

Note that one can easily adapt this lower bound to variants where $B$ is some other shape, such as a square, a triangle or a disk.

As usual, we can adapt the proof to the incremental case and show the following.
\begin{lemma}
For any incremental data structure $D$ for Largest Empty Disk in a Set of Disks with per-operation runtime guarantees, Scenario 1 in Theorem \ref{thm:general_3sum} applies (with $t_{uq} = t_u + t_q$).
\end{lemma}

We can give an upper bound as follows using the techniques of Chan \cite{Chan2020}. Map each input disk with center $(a,b)\in \R^2$ and radius $r$ to the plane in $\R^3$ of equation $z=-2ax-2by + a^2 + b^2-r$. Add 4 near-vertical planes along the edges of $B$. Then finding the largest empty disk with center in $B$ reduces to finding the vertex of the lower envelope of these planes which maximizes $x^2+y^2+z$. Chan gives a data structure for this problem with $O(n^{11/12+\epsilon})$ amortized update time.
In the more restricted semi-online setting (which generalizes the incremental case), another paper by the same author \cite{Chan2003} gives a data structure with $O(n^{7/8+\epsilon})$ worst-case time per update.

\subsection{Rectangle Covering with Disks}\label{subsec:rectangle_cover_with_disks}

\subparagraph*{Rectangle Covering with Disks:}Given some fixed rectangle in the plane $B$, maintain a set of $O(n)$ disks $S$ together with the answer to the answer to the question ``is $R$ contained in the union of all disks in $S$?''.

\begin{lemma}
For any fully-dynamic data structure $D$ for Rectangle Covering with Disks, Scenario 2 in Theorem \ref{thm:general_3sum} applies (with $t_{uq} = t_u + t_q$).
\end{lemma}
\begin{proof}[Proof Sketch]
We can use almost the same proof as for Largest Empty Disk in a Set of Disks, where we increase the squared radius of every disk by $25k^2 + 4$ (additively). Then there is a point in $B$ which is not covered by these disks with increased radius if and only if there is an empty disk of radius $25k^2 +4$ in the original set of disks. The rest of the proof follows similarly.
\end{proof}

Again, one can easily adapt this lower bound to variants where $B$ is some other shape to cover, such as a square, a triangle or a disk.

Adapting the proof to the incremental setting gives the following.
\begin{lemma}
For any incremental data structure $D$ for Rectangle Covering with Disks with per-operation runtime guarantees, Scenario 1 in Theorem \ref{thm:general_3sum} applies (with $t_{uq} = t_u + t_q$).
\end{lemma}

As this problem reduces easily to Largest Empty Disk in a Set of Disks, we get the same upper bounds here.

\subsection{Convex Layer Size in \texorpdfstring{$\R^2$}{R2}}\label{subsec:convex_layer_size}
\subparagraph*{Convex Layer Size in $\R^2$:} Maintain a set $S$ of $n$ points in the plane to support queries asking for the number of vertices on the $i$'th convex layer (for any $1\leq i \leq n$).

\begin{lemma}\label{thm:convex_layers}
For any fully-dynamic data structure $D$ for Convex Layer Size in $\R^2$, Scenario 2 in Theorem \ref{thm:general_3sum} applies (with $t_{uq} = t_q$).
\end{lemma}

\begin{figure}
    \centering
    \includegraphics[scale=0.59]{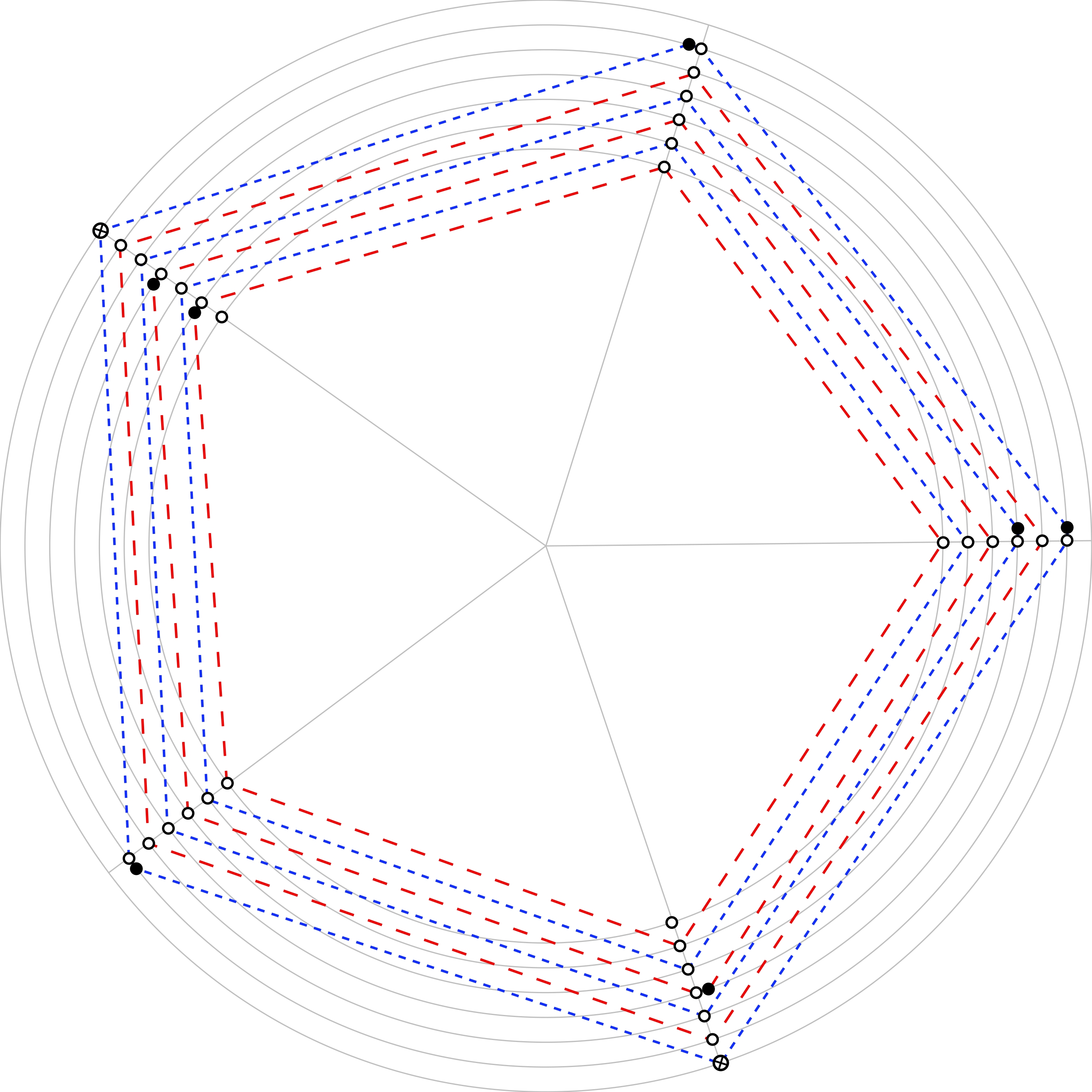}
    \caption{Illustration (not to scale) of a set of points obtained after Step 2 in the proof of Theorem \ref{thm:convex_layers}. The illustrated instance has $m=5$, $k=3$, $\mathcal{F} = \{\{1,3,5\}, \{2,4,5\}, \{2\} \}$ and $J = \{1,3,5\}$. The points $q_{\bullet, \bullet}$ and $q'_{\bullet, \bullet}$ are represented in white, the points $p_{\bullet, \bullet}$ are represented in black. The points $b_{\bullet}$ are represented with a cross in a white circle. The convex layers are represented by dotted lines (except for the inner most degenerate layer consisting of two points).}
    \label{fig:layers}
\end{figure}

\begin{proof}[Proof Sketch]
Let $\mathcal{F}=\{F_1, \ldots, F_k\}$ be a family of $k$ subsets of $\{1,2,\ldots m\}$. Suppose without loss of generality that $m\geq 3$.

Here it will be more convenient to work in polar coordinates $(r, \theta)$ (where this for example would denote the point $(r\cos\theta, r\sin\theta)$ in Cartesian coordinates). We will stick with this convention for the whole proof. 

Let $\epsilon = \frac{1}{km^2}$ and $\alpha = \frac{1}{km^3}$. We perform Step 1 by initializing $D$ with all the following points.
\begin{itemize}
    \item For all $1 \leq i \leq k$ and $1\leq j \leq m$, put a point $q_{i,j}$ at polar coordinates $(1+2(k-i)\epsilon, \frac{2\pi}{m}j)$ and a point $q'_{i,j}$ at polar coordinates $(1+2(k-i)\epsilon + \epsilon, \frac{2\pi}{m}j)$.
    \item For all $1 \leq i \leq k$ and $1\leq j \leq m$ such that $j\in F_i$ put a point $p_{i,j}$ at polar coordinates $(1+2(k-i)\epsilon + \epsilon, \frac{2\pi}{m}j+\alpha)$.
    \item For all $1\leq j \leq m$, put a point $b_j$ at polar coordinates $(1+2k\epsilon , \frac{2\pi}{m}j)$.
\end{itemize}

The total number of points is $O(m\cdot k)$.

To perform Step 2 when given $J \subset \{1,2,\ldots m\}$, we delete the points $b_{j}$ for all $j\in J$. This requires $O(|J|)$ updates on $D$. See Figure \ref{fig:layers} for an illustration.

At this point, assuming $J \neq \{1,2,\ldots m\}$, the set of points $S$ has exactly $2k+1$ convex layers. Every convex layer except for the last (the innermost) has $m+1$ vertices which are not of the form $p_{\bullet, \bullet}$. Moreover, every point of the type $p_{i,j}$ appears on layer $2i-1$ if $b_j \not\in S$ (i.e.\  $j \in J$) or layer $2i$ if $b_j \in S$ (i.e.\  $j \not\in J$). Thus, layer $2i-1$ is of size $m+1$ if and only if there is no point $p_{i,j}$ such that $j \in S$. In other words, layer $2i-1$ is of size $m+1$ if and only if $J\cap F_i = \emptyset$. We can thus answer an intersection query in Step 3 by a single query to $D$.

Just as for the proof of Theorem \ref{thm:extremal_pts_ds}, one can show that this still works if we limit the precision of the coordinates to $O(\log n)$ bits of precision. 
\end{proof}

Adapting the proof to the incremental setting gives the following.
\begin{lemma}
For any incremental data structure $D$ for Convex Layer Size in $\R^2$ with per-operation runtime guarantees, Scenario 1 in Theorem \ref{thm:general_3sum} applies (with $t_{uq} = t_q$).
\end{lemma}

For this problem we could use a similar approach to the one we used for the Depth Problem with Squares to show lower bounds for incremental data structures with amortized runtime, conditioned on the hardness of the OMv problem (instead of deleting the points $b_j$ to go back to the state of the data structure before Step 3, we continue adding new convex layers and adapt the queries in Step 3 accordingly).

Note that in the static case, Chazelle \cite{Chazelle1985} showed how to compute all convex layers in $O(n\log n)$ time.

\bibliography{refs}

\begin{thebibliography}{10}

\bibitem{AbboudD2016}
Amir Abboud and S{\o}ren Dahlgaard.
\newblock Popular conjectures as a barrier for dynamic planar graph algorithms.
\newblock In Irit Dinur, editor, {\em {IEEE} 57th Annual Symposium on
  Foundations of Computer Science, {FOCS} 2016, 9-11 October 2016, Hyatt
  Regency, New Brunswick, New Jersey, {USA}}, pages 477--486. {IEEE} Computer
  Society, 2016.
\newblock \href {https://doi.org/10.1109/FOCS.2016.58}
  {\path{doi:10.1109/FOCS.2016.58}}.

\bibitem{AbboudW2014}
Amir Abboud and Virginia~Vassilevska Williams.
\newblock Popular conjectures imply strong lower bounds for dynamic problems.
\newblock In {\em 55th {IEEE} Annual Symposium on Foundations of Computer
  Science, {FOCS} 2014, Philadelphia, PA, USA, October 18-21, 2014}, pages
  434--443. {IEEE} Computer Society, 2014.
\newblock \href {https://doi.org/10.1109/FOCS.2014.53}
  {\path{doi:10.1109/FOCS.2014.53}}.

\bibitem{AbboudWY2018}
Amir Abboud, Virginia~Vassilevska Williams, and Huacheng Yu.
\newblock Matching triangles and basing hardness on an extremely popular
  conjecture.
\newblock {\em {SIAM} J. Comput.}, 47(3):1098--1122, 2018.
\newblock \href {https://doi.org/10.1137/15M1050987}
  {\path{doi:10.1137/15M1050987}}.

\bibitem{AgarwalGM2002_coloredSearching}
Pankaj~K. Agarwal, Sathish Govindarajan, and S.~Muthukrishnan.
\newblock Range searching in categorical data: Colored range searching on grid.
\newblock In Rolf~H. M{\"{o}}hring and Rajeev Raman, editors, {\em Algorithms -
  {ESA} 2002, 10th Annual European Symposium, Rome, Italy, September 17-21,
  2002, Proceedings}, volume 2461 of {\em Lecture Notes in Computer Science},
  pages 17--28. Springer, 2002.
\newblock \href {https://doi.org/10.1007/3-540-45749-6\_6}
  {\path{doi:10.1007/3-540-45749-6\_6}}.

\bibitem{AlmanMW2017_finegrained}
Josh Alman, Matthias Mnich, and Virginia~Vassilevska Williams.
\newblock Dynamic parameterized problems and algorithms.
\newblock In Ioannis Chatzigiannakis, Piotr Indyk, Fabian Kuhn, and Anca
  Muscholl, editors, {\em 44th International Colloquium on Automata, Languages,
  and Programming, {ICALP} 2017, July 10-14, 2017, Warsaw, Poland}, volume~80
  of {\em LIPIcs}, pages 41:1--41:16. Schloss Dagstuhl - Leibniz-Zentrum
  f{\"{u}}r Informatik, 2017.
\newblock \href {https://doi.org/10.4230/LIPIcs.ICALP.2017.41}
  {\path{doi:10.4230/LIPIcs.ICALP.2017.41}}.

\bibitem{AmirCLL2014_finegrained}
Amihood Amir, Timothy~M. Chan, Moshe Lewenstein, and Noa Lewenstein.
\newblock On hardness of jumbled indexing.
\newblock In Javier Esparza, Pierre Fraigniaud, Thore Husfeldt, and Elias
  Koutsoupias, editors, {\em Automata, Languages, and Programming - 41st
  International Colloquium, {ICALP} 2014, Copenhagen, Denmark, July 8-11, 2014,
  Proceedings, Part {I}}, volume 8572 of {\em Lecture Notes in Computer
  Science}, pages 114--125. Springer, 2014.
\newblock \href {https://doi.org/10.1007/978-3-662-43948-7\_10}
  {\path{doi:10.1007/978-3-662-43948-7\_10}}.

\bibitem{AmirKLPPS2019-finegrained}
Amihood Amir, Tsvi Kopelowitz, Avivit Levy, Seth Pettie, Ely Porat, and B.~Riva
  Shalom.
\newblock Mind the gap! - online dictionary matching with one gap.
\newblock {\em Algorithmica}, 81(6):2123--2157, 2019.
\newblock \href {https://doi.org/10.1007/s00453-018-0526-2}
  {\path{doi:10.1007/s00453-018-0526-2}}.

\bibitem{Baran2008}
Ilya Baran, Erik~D. Demaine, and Mihai Patrascu.
\newblock Subquadratic algorithms for {3SUM}.
\newblock {\em Algorithmica}, 50(4):584--596, 2008.
\newblock \href {https://doi.org/10.1007/s00453-007-9036-3}
  {\path{doi:10.1007/s00453-007-9036-3}}.

\bibitem{BaswanaCC02016-finegrained}
Surender Baswana, Shreejit~Ray Chaudhury, Keerti Choudhary, and Shahbaz Khan.
\newblock Dynamic {DFS} in undirected graphs: breaking the o(\emph{m}) barrier.
\newblock In Robert Krauthgamer, editor, {\em Proceedings of the Twenty-Seventh
  Annual {ACM-SIAM} Symposium on Discrete Algorithms, {SODA} 2016, Arlington,
  VA, USA, January 10-12, 2016}, pages 730--739. {SIAM}, 2016.
\newblock \href {https://doi.org/10.1137/1.9781611974331.ch52}
  {\path{doi:10.1137/1.9781611974331.ch52}}.

\bibitem{Bentley2015}
Jon~Louis Bentley.
\newblock Multidimensional binary search trees used for associative searching.
\newblock {\em Commun. {ACM}}, 18(9):509--517, 1975.
\newblock URL: \url{http://doi.acm.org/10.1145/361002.361007}, \href
  {https://doi.org/10.1145/361002.361007} {\path{doi:10.1145/361002.361007}}.

\bibitem{BerkholzKS2017-finegrained}
Christoph Berkholz, Jens Keppeler, and Nicole Schweikardt.
\newblock Answering conjunctive queries under updates.
\newblock In Emanuel Sallinger, Jan~Van den Bussche, and Floris Geerts,
  editors, {\em Proceedings of the 36th {ACM} {SIGMOD-SIGACT-SIGAI} Symposium
  on Principles of Database Systems, {PODS} 2017, Chicago, IL, USA, May 14-19,
  2017}, pages 303--318. {ACM}, 2017.
\newblock \href {https://doi.org/10.1145/3034786.3034789}
  {\path{doi:10.1145/3034786.3034789}}.

\bibitem{BerkholzKS2018-finegrained}
Christoph Berkholz, Jens Keppeler, and Nicole Schweikardt.
\newblock Answering {UCQ}s under updates and in the presence of integrity
  constraints.
\newblock In Benny Kimelfeld and Yael Amsterdamer, editors, {\em 21st
  International Conference on Database Theory, {ICDT} 2018, March 26-29, 2018,
  Vienna, Austria}, volume~98 of {\em LIPIcs}, pages 8:1--8:19. Schloss
  Dagstuhl - Leibniz-Zentrum f{\"{u}}r Informatik, 2018.
\newblock \href {https://doi.org/10.4230/LIPIcs.ICDT.2018.8}
  {\path{doi:10.4230/LIPIcs.ICDT.2018.8}}.

\bibitem{Bringmann2013}
Karl Bringmann.
\newblock Bringing order to special cases of {K}lee's measure problem.
\newblock In Krishnendu Chatterjee and Jir{\'{\i}} Sgall, editors, {\em
  Mathematical Foundations of Computer Science 2013 - 38th International
  Symposium, {MFCS} 2013, Klosterneuburg, Austria, August 26-30, 2013.
  Proceedings}, volume 8087 of {\em Lecture Notes in Computer Science}, pages
  207--218. Springer, 2013.
\newblock \href {https://doi.org/10.1007/978-3-642-40313-2\_20}
  {\path{doi:10.1007/978-3-642-40313-2\_20}}.

\bibitem{Bringmann2019}
Karl Bringmann.
\newblock Fine-grained complexity theory (tutorial).
\newblock In Rolf Niedermeier and Christophe Paul, editors, {\em 36th
  International Symposium on Theoretical Aspects of Computer Science, {STACS}
  2019, March 13-16, 2019, Berlin, Germany}, volume 126 of {\em LIPIcs}, pages
  4:1--4:7. Schloss Dagstuhl - Leibniz-Zentrum f{\"{u}}r Informatik, 2019.
\newblock \href {https://doi.org/10.4230/LIPIcs.STACS.2019.4}
  {\path{doi:10.4230/LIPIcs.STACS.2019.4}}.

\bibitem{Bringmann2021}
Karl Bringmann.
\newblock Fine-grained complexity theory: Conditional lower bounds for
  computational geometry.
\newblock In Liesbeth~De Mol, Andreas Weiermann, Florin Manea, and David
  Fern{\'{a}}ndez{-}Duque, editors, {\em Connecting with Computability - 17th
  Conference on Computability in Europe, CiE 2021, Virtual Event, Ghent, July
  5-9, 2021, Proceedings}, volume 12813 of {\em Lecture Notes in Computer
  Science}, pages 60--70. Springer, 2021.
\newblock \href {https://doi.org/10.1007/978-3-030-80049-9\_6}
  {\path{doi:10.1007/978-3-030-80049-9\_6}}.

\bibitem{CardinalIK2021}
Jean Cardinal, John Iacono, and Grigorios Koumoutsos.
\newblock Worst-case efficient dynamic geometric independent set.
\newblock In Petra Mutzel, Rasmus Pagh, and Grzegorz Herman, editors, {\em 29th
  Annual European Symposium on Algorithms, {ESA} 2021, September 6-8, 2021,
  Lisbon, Portugal (Virtual Conference)}, volume 204 of {\em LIPIcs}, pages
  25:1--25:15. Schloss Dagstuhl - Leibniz-Zentrum f{\"{u}}r Informatik, 2021.
\newblock \href {https://doi.org/10.4230/LIPIcs.ESA.2021.25}
  {\path{doi:10.4230/LIPIcs.ESA.2021.25}}.

\bibitem{Chan2003}
Timothy~M. Chan.
\newblock Semi-online maintenance of geometric optima and measures.
\newblock {\em {SIAM} J. Comput.}, 32(3):700--716, 2003.
\newblock \href {https://doi.org/10.1137/S0097539702404389}
  {\path{doi:10.1137/S0097539702404389}}.

\bibitem{Chan2010}
Timothy~M. Chan.
\newblock A (slightly) faster algorithm for {K}lee's measure problem.
\newblock {\em Computational Geometry}, 43(3):243--250, 2010.
\newblock Special Issue on 24th Annual Symposium on Computational Geometry
  (SoCG'08).
\newblock URL:
  \url{https://www.sciencedirect.com/science/article/pii/S0925772109000595},
  \href {https://doi.org/https://doi.org/10.1016/j.comgeo.2009.01.007}
  {\path{doi:https://doi.org/10.1016/j.comgeo.2009.01.007}}.

\bibitem{Chan2020}
Timothy~M. Chan.
\newblock Dynamic geometric data structures via shallow cuttings.
\newblock {\em Discret. Comput. Geom.}, 64(4):1235--1252, 2020.
\newblock \href {https://doi.org/10.1007/s00454-020-00229-5}
  {\path{doi:10.1007/s00454-020-00229-5}}.

\bibitem{Chan2020-3SUM}
Timothy~M. Chan.
\newblock More logarithmic-factor speedups for {3SUM}, (median, +)-convolution,
  and some geometric {3SUM}-hard problems.
\newblock {\em {ACM} Trans. Algorithms}, 16(1):7:1--7:23, 2020.
\newblock \href {https://doi.org/10.1145/3363541} {\path{doi:10.1145/3363541}}.

\bibitem{Chan2022}
Timothy~M. Chan, Qizheng He, Subhash Suri, and Jie Xue.
\newblock Dynamic geometric set cover, revisited.
\newblock In {\em Proceedings of the 2022 Annual ACM-SIAM Symposium on Discrete
  Algorithms (SODA)}, pages 3496--3528, 2022.
\newblock URL:
  \url{https://epubs.siam.org/doi/abs/10.1137/1.9781611977073.139}, \href
  {https://doi.org/10.1137/1.9781611977073.139}
  {\path{doi:10.1137/1.9781611977073.139}}.

\bibitem{ChanH2021_coloredSearching}
Timothy~M. Chan and Zhengcheng Huang.
\newblock Dynamic colored orthogonal range searching.
\newblock In Petra Mutzel, Rasmus Pagh, and Grzegorz Herman, editors, {\em 29th
  Annual European Symposium on Algorithms, {ESA} 2021, September 6-8, 2021,
  Lisbon, Portugal (Virtual Conference)}, volume 204 of {\em LIPIcs}, pages
  28:1--28:13. Schloss Dagstuhl - Leibniz-Zentrum f{\"{u}}r Informatik, 2021.
\newblock \href {https://doi.org/10.4230/LIPIcs.ESA.2021.28}
  {\path{doi:10.4230/LIPIcs.ESA.2021.28}}.

\bibitem{ChanN2020_coloredSearching}
Timothy~M. Chan and Yakov Nekrich.
\newblock Better data structures for colored orthogonal range reporting.
\newblock In Shuchi Chawla, editor, {\em Proceedings of the 2020 {ACM-SIAM}
  Symposium on Discrete Algorithms, {SODA} 2020, Salt Lake City, UT, USA,
  January 5-8, 2020}, pages 627--636. {SIAM}, 2020.
\newblock \href {https://doi.org/10.1137/1.9781611975994.38}
  {\path{doi:10.1137/1.9781611975994.38}}.

\bibitem{ChanWX2022}
Timothy~M. Chan, Virginia~Vassilevska Williams, and Yinzhan Xu.
\newblock Hardness for triangle problems under even more believable hypotheses:
  reductions from real {APSP}, real {3SUM}, and {OV}.
\newblock In Stefano Leonardi and Anupam Gupta, editors, {\em {STOC} '22: 54th
  Annual {ACM} {SIGACT} Symposium on Theory of Computing, Rome, Italy, June 20
  - 24, 2022}, pages 1501--1514. {ACM}, 2022.
\newblock \href {https://doi.org/10.1145/3519935.3520032}
  {\path{doi:10.1145/3519935.3520032}}.

\bibitem{ChanHopcroft2022}
Timothy~M. Chan and Da~Wei Zheng.
\newblock Hopcroft's problem, log-star shaving, 2d fractional cascading, and
  decision trees.
\newblock In {\em Proceedings of the 2022 Annual ACM-SIAM Symposium on Discrete
  Algorithms (SODA)}, pages 190--210, 2022.
\newblock URL: \url{https://epubs.siam.org/doi/abs/10.1137/1.9781611977073.10},
  \href {https://doi.org/10.1137/1.9781611977073.10}
  {\path{doi:10.1137/1.9781611977073.10}}.

\bibitem{Chazelle1985}
Bernard Chazelle.
\newblock On the convex layers of a planar set.
\newblock {\em {IEEE} Trans. Inf. Theory}, 31(4):509--517, 1985.
\newblock \href {https://doi.org/10.1109/TIT.1985.1057060}
  {\path{doi:10.1109/TIT.1985.1057060}}.

\bibitem{ChenDGWXY2018-finegrained}
Lijie Chen, Erik~D. Demaine, Yuzhou Gu, Virginia~Vassilevska Williams, Yinzhan
  Xu, and Yuancheng Yu.
\newblock Nearly optimal separation between partially and fully retroactive
  data structures.
\newblock In David Eppstein, editor, {\em 16th Scandinavian Symposium and
  Workshops on Algorithm Theory, {SWAT} 2018, June 18-20, 2018, Malm{\"{o}},
  Sweden}, volume 101 of {\em LIPIcs}, pages 33:1--33:12. Schloss Dagstuhl -
  Leibniz-Zentrum f{\"{u}}r Informatik, 2018.
\newblock \href {https://doi.org/10.4230/LIPIcs.SWAT.2018.33}
  {\path{doi:10.4230/LIPIcs.SWAT.2018.33}}.

\bibitem{CLRS}
Thomas~H. Cormen, Charles~E. Leiserson, Ronald~L. Rivest, and Clifford Stein.
\newblock {\em Introduction to Algorithms, 3rd Edition}.
\newblock {MIT} Press, 2009.
\newblock URL: \url{http://mitpress.mit.edu/books/introduction-algorithms}.

\bibitem{Dahlgaard2016_finegrained}
S{\o}ren Dahlgaard.
\newblock On the hardness of partially dynamic graph problems and connections
  to diameter.
\newblock In Ioannis Chatzigiannakis, Michael Mitzenmacher, Yuval Rabani, and
  Davide Sangiorgi, editors, {\em 43rd International Colloquium on Automata,
  Languages, and Programming, {ICALP} 2016, July 11-15, 2016, Rome, Italy},
  volume~55 of {\em LIPIcs}, pages 48:1--48:14. Schloss Dagstuhl -
  Leibniz-Zentrum f{\"{u}}r Informatik, 2016.
\newblock \href {https://doi.org/10.4230/LIPIcs.ICALP.2016.48}
  {\path{doi:10.4230/LIPIcs.ICALP.2016.48}}.

\bibitem{Fowler1981}
Robert~J. Fowler, Mike Paterson, and Steven~L. Tanimoto.
\newblock Optimal packing and covering in the plane are np-complete.
\newblock {\em Inf. Process. Lett.}, 12(3):133--137, 1981.
\newblock \href {https://doi.org/10.1016/0020-0190(81)90111-3}
  {\path{doi:10.1016/0020-0190(81)90111-3}}.

\bibitem{Frandsen2001}
Gudmund~Skovbjerg Frandsen, Johan~P Hansen, and Peter~Bro Miltersen.
\newblock Lower bounds for dynamic algebraic problems.
\newblock {\em Information and Computation}, 171(2):333--349, 2001.
\newblock URL:
  \url{https://www.sciencedirect.com/science/article/pii/S0890540101930469},
  \href {https://doi.org/https://doi.org/10.1006/inco.2001.3046}
  {\path{doi:https://doi.org/10.1006/inco.2001.3046}}.

\bibitem{Gajentaan1995}
Anka Gajentaan and Mark~H. Overmars.
\newblock On a class of o(n2) problems in computational geometry.
\newblock {\em Comput. Geom.}, 5:165--185, 1995.
\newblock \href {https://doi.org/10.1016/0925-7721(95)00022-2}
  {\path{doi:10.1016/0925-7721(95)00022-2}}.

\bibitem{Gronlund2014}
Allan Gr{\o}nlund and Seth Pettie.
\newblock Threesomes, degenerates, and love triangles.
\newblock In {\em 55th {IEEE} Annual Symposium on Foundations of Computer
  Science, {FOCS} 2014, Philadelphia, PA, USA, October 18-21, 2014}, pages
  621--630. {IEEE} Computer Society, 2014.
\newblock \href {https://doi.org/10.1109/FOCS.2014.72}
  {\path{doi:10.1109/FOCS.2014.72}}.

\bibitem{GuerreiroFP2021}
Andreia~P. Guerreiro, Carlos~M. Fonseca, and Lu{\'{\i}}s Paquete.
\newblock The hypervolume indicator: Computational problems and algorithms.
\newblock {\em {ACM} Comput. Surv.}, 54(6):119:1--119:42, 2021.
\newblock \href {https://doi.org/10.1145/3453474} {\path{doi:10.1145/3453474}}.

\bibitem{Gupta2018_coloredSearching}
Prosenjit Gupta, Ravi Janardan, Saladi Rahul, and Michiel H.~M. Smid.
\newblock Computational geometry: Generalized (or colored) intersection
  searching.
\newblock In Dinesh~P. Mehta and Sartaj Sahni, editors, {\em Handbook of Data
  Structures and Applications}, chapter~67, page 1042–1057. CRC Press, 2nd
  edition, 2018.
\newblock URL:
  \url{https://www-users.cs.umn.edu/~sala0198/Papers/ds2-handbook.pdf}.

\bibitem{GuptaJS1995_coloredSearching}
Prosenjit Gupta, Ravi Janardan, and Michiel H.~M. Smid.
\newblock Further results on generalized intersection searching problems:
  Counting, reporting, and dynamization.
\newblock {\em J. Algorithms}, 19(2):282--317, 1995.
\newblock \href {https://doi.org/10.1006/jagm.1995.1038}
  {\path{doi:10.1006/jagm.1995.1038}}.

\bibitem{Henzinger2015}
Monika Henzinger, Sebastian Krinninger, Danupon Nanongkai, and Thatchaphol
  Saranurak.
\newblock Unifying and strengthening hardness for dynamic problems via the
  online matrix-vector multiplication conjecture.
\newblock In Rocco~A. Servedio and Ronitt Rubinfeld, editors, {\em Proceedings
  of the Forty-Seventh Annual {ACM} on Symposium on Theory of Computing, {STOC}
  2015, Portland, OR, USA, June 14-17, 2015}, pages 21--30. {ACM}, 2015.
\newblock \href {https://doi.org/10.1145/2746539.2746609}
  {\path{doi:10.1145/2746539.2746609}}.

\bibitem{HenzingerL0W2017_finegrained}
Monika Henzinger, Andrea Lincoln, Stefan Neumann, and Virginia~Vassilevska
  Williams.
\newblock Conditional hardness for sensitivity problems.
\newblock In Christos~H. Papadimitriou, editor, {\em 8th Innovations in
  Theoretical Computer Science Conference, {ITCS} 2017, January 9-11, 2017,
  Berkeley, CA, {USA}}, volume~67 of {\em LIPIcs}, pages 26:1--26:31. Schloss
  Dagstuhl - Leibniz-Zentrum f{\"{u}}r Informatik, 2017.
\newblock \href {https://doi.org/10.4230/LIPIcs.ITCS.2017.26}
  {\path{doi:10.4230/LIPIcs.ITCS.2017.26}}.

\bibitem{JanardanL1993_coloredSearching}
Ravi Janardan and Mario~Alberto L{\'{o}}pez.
\newblock Generalized intersection searching problems.
\newblock {\em Int. J. Comput. Geom. Appl.}, 3(1):39--69, 1993.
\newblock \href {https://doi.org/10.1142/S021819599300004X}
  {\path{doi:10.1142/S021819599300004X}}.

\bibitem{Jin2022}
Ce~Jin and Yinzhan Xu.
\newblock Tight dynamic problem lower bounds from generalized {BMM} and {OMv}.
\newblock {\em To appear in {STOC'22}. CoRR}, abs/2202.11250, 2022.
\newblock URL: \url{https://arxiv.org/abs/2202.11250}, \href
  {http://arxiv.org/abs/2202.11250} {\path{arXiv:2202.11250}}.

\bibitem{KarczmarzL2015_finegrained}
Adam Karczmarz and Jakub Lacki.
\newblock Fast and simple connectivity in graph timelines.
\newblock In Frank Dehne, J{\"{o}}rg{-}R{\"{u}}diger Sack, and Ulrike Stege,
  editors, {\em Algorithms and Data Structures - 14th International Symposium,
  {WADS} 2015, Victoria, BC, Canada, August 5-7, 2015. Proceedings}, volume
  9214 of {\em Lecture Notes in Computer Science}, pages 458--469. Springer,
  2015.
\newblock \href {https://doi.org/10.1007/978-3-319-21840-3\_38}
  {\path{doi:10.1007/978-3-319-21840-3\_38}}.

\bibitem{Ko2020}
Young~Kun Ko and Min~Jae Song.
\newblock Hardness of approximate nearest neighbor search under l-infinity.
\newblock {\em CoRR}, abs/2011.06135, 2020.
\newblock URL: \url{https://arxiv.org/abs/2011.06135}, \href
  {http://arxiv.org/abs/2011.06135} {\path{arXiv:2011.06135}}.

\bibitem{KopelowitzK2016-finegrained}
Tsvi Kopelowitz and Robert Krauthgamer.
\newblock Color-distance oracles and snippets.
\newblock In Roberto Grossi and Moshe Lewenstein, editors, {\em 27th Annual
  Symposium on Combinatorial Pattern Matching, {CPM} 2016, June 27-29, 2016,
  Tel Aviv, Israel}, volume~54 of {\em LIPIcs}, pages 24:1--24:10. Schloss
  Dagstuhl - Leibniz-Zentrum f{\"{u}}r Informatik, 2016.
\newblock \href {https://doi.org/10.4230/LIPIcs.CPM.2016.24}
  {\path{doi:10.4230/LIPIcs.CPM.2016.24}}.

\bibitem{Kopelowitz2016}
Tsvi Kopelowitz, Seth Pettie, and Ely Porat.
\newblock Higher lower bounds from the {3SUM} conjecture.
\newblock In Robert Krauthgamer, editor, {\em Proceedings of the Twenty-Seventh
  Annual {ACM-SIAM} Symposium on Discrete Algorithms, {SODA} 2016, Arlington,
  VA, USA, January 10-12, 2016}, pages 1272--1287. {SIAM}, 2016.
\newblock \href {https://doi.org/10.1137/1.9781611974331.ch89}
  {\path{doi:10.1137/1.9781611974331.ch89}}.

\bibitem{LarsenW2013_coloredSearching}
Kasper~Green Larsen and Freek van Walderveen.
\newblock Near-optimal range reporting structures for categorical data.
\newblock In Sanjeev Khanna, editor, {\em Proceedings of the Twenty-Fourth
  Annual {ACM-SIAM} Symposium on Discrete Algorithms, {SODA} 2013, New Orleans,
  Louisiana, USA, January 6-8, 2013}, pages 265--276. {SIAM}, 2013.
\newblock \href {https://doi.org/10.1137/1.9781611973105.20}
  {\path{doi:10.1137/1.9781611973105.20}}.

\bibitem{Larsen2017}
Kasper~Green Larsen and R.~Ryan Williams.
\newblock Faster online matrix-vector multiplication.
\newblock In Philip~N. Klein, editor, {\em Proceedings of the Twenty-Eighth
  Annual {ACM-SIAM} Symposium on Discrete Algorithms, {SODA} 2017, Barcelona,
  Spain, Hotel Porta Fira, January 16-19}, pages 2182--2189. {SIAM}, 2017.
\newblock \href {https://doi.org/10.1137/1.9781611974782.142}
  {\path{doi:10.1137/1.9781611974782.142}}.

\bibitem{LauR2021}
Joshua Lau and Angus Ritossa.
\newblock Algorithms and hardness for multidimensional range updates and
  queries.
\newblock In James~R. Lee, editor, {\em 12th Innovations in Theoretical
  Computer Science Conference, {ITCS} 2021, January 6-8, 2021, Virtual
  Conference}, volume 185 of {\em LIPIcs}, pages 35:1--35:20. Schloss Dagstuhl
  - Leibniz-Zentrum f{\"{u}}r Informatik, 2021.
\newblock \href {https://doi.org/10.4230/LIPIcs.ITCS.2021.35}
  {\path{doi:10.4230/LIPIcs.ITCS.2021.35}}.

\bibitem{Mortensen2003_coloredSearching}
Christian~W. Mortensen.
\newblock Generalized static orthogonal range searching in less space.
\newblock Technical Report TR-2003-33, IT University of Copenhagen, Copenhagen,
  Denmark, September 2003.

\bibitem{Nekrich2014_coloredSearching}
Yakov Nekrich.
\newblock Efficient range searching for categorical and plain data.
\newblock {\em {ACM} Trans. Database Syst.}, 39(1):9:1--9:21, 2014.
\newblock \href {https://doi.org/10.1145/2543924} {\path{doi:10.1145/2543924}}.

\bibitem{OvermarsPast1981}
Mark~H. Overmars.
\newblock Searching in the past {II}: General transforms.
\newblock Technical Report RUU-CS-81-9, Department of Computer Science,
  University of Utrecht, Utrecht, The Netherlands, May 1981.

\bibitem{Overmars1981}
Mark~H. Overmars and Jan van Leeuwen.
\newblock Maintenance of configurations in the plane.
\newblock {\em J. Comput. Syst. Sci.}, 23(2):166--204, 1981.
\newblock \href {https://doi.org/10.1016/0022-0000(81)90012-X}
  {\path{doi:10.1016/0022-0000(81)90012-X}}.

\bibitem{Overmars1988}
M.H. Overmars and Chee-Keng Yap.
\newblock New upper bounds in {K}lee's measure problem.
\newblock In {\em [Proceedings 1988] 29th Annual Symposium on Foundations of
  Computer Science}, pages 550--556, 1988.
\newblock \href {https://doi.org/10.1109/SFCS.1988.21971}
  {\path{doi:10.1109/SFCS.1988.21971}}.

\bibitem{Probst2018-finegrained}
Maximilian Probst.
\newblock On the complexity of the (approximate) nearest colored node problem.
\newblock In Yossi Azar, Hannah Bast, and Grzegorz Herman, editors, {\em 26th
  Annual European Symposium on Algorithms, {ESA} 2018, August 20-22, 2018,
  Helsinki, Finland}, volume 112 of {\em LIPIcs}, pages 68:1--68:14. Schloss
  Dagstuhl - Leibniz-Zentrum f{\"{u}}r Informatik, 2018.
\newblock \href {https://doi.org/10.4230/LIPIcs.ESA.2018.68}
  {\path{doi:10.4230/LIPIcs.ESA.2018.68}}.

\bibitem{Patrascu2010}
Mihai P\u{a}tra\c{s}cu.
\newblock Towards polynomial lower bounds for dynamic problems.
\newblock In {\em Proceedings of the Forty-Second ACM Symposium on Theory of
  Computing}, STOC '10, page 603–610, New York, NY, USA, 2010. Association
  for Computing Machinery.
\newblock \href {https://doi.org/10.1145/1806689.1806772}
  {\path{doi:10.1145/1806689.1806772}}.

\bibitem{Rubinstein2018}
Aviad Rubinstein.
\newblock Hardness of approximate nearest neighbor search.
\newblock In Ilias Diakonikolas, David Kempe, and Monika Henzinger, editors,
  {\em Proceedings of the 50th Annual {ACM} {SIGACT} Symposium on Theory of
  Computing, {STOC} 2018, Los Angeles, CA, USA, June 25-29, 2018}, pages
  1260--1268. {ACM}, 2018.
\newblock \href {https://doi.org/10.1145/3188745.3188916}
  {\path{doi:10.1145/3188745.3188916}}.

\bibitem{ShiJ2005_coloredSearching}
Qingmin Shi and Joseph~F. J{\'{a}}J{\'{a}}.
\newblock Optimal and near-optimal algorithms for generalized intersection
  reporting on pointer machines.
\newblock {\em Inf. Process. Lett.}, 95(3):382--388, 2005.
\newblock \href {https://doi.org/10.1016/j.ipl.2005.04.008}
  {\path{doi:10.1016/j.ipl.2005.04.008}}.

\bibitem{BrandNS2019-finegrained}
Jan van~den Brand, Danupon Nanongkai, and Thatchaphol Saranurak.
\newblock Dynamic matrix inverse: Improved algorithms and matching conditional
  lower bounds.
\newblock In David Zuckerman, editor, {\em 60th {IEEE} Annual Symposium on
  Foundations of Computer Science, {FOCS} 2019, Baltimore, Maryland, USA,
  November 9-12, 2019}, pages 456--480. {IEEE} Computer Society, 2019.
\newblock \href {https://doi.org/10.1109/FOCS.2019.00036}
  {\path{doi:10.1109/FOCS.2019.00036}}.

\bibitem{Williams2018}
R.~Ryan Williams.
\newblock Faster all-pairs shortest paths via circuit complexity.
\newblock {\em {SIAM} J. Comput.}, 47(5):1965--1985, 2018.
\newblock \href {https://doi.org/10.1137/15M1024524}
  {\path{doi:10.1137/15M1024524}}.

\bibitem{Williams2019}
Virginia~Vassilevska Williams.
\newblock {\em On some fine-grained questions in algorithms and complexity},
  pages 3447--3487.
\newblock WORLD SCIENTIFIC, 2019.
\newblock URL:
  \url{https://www.worldscientific.com/doi/abs/10.1142/9789813272880_0188},
  \href
  {http://arxiv.org/abs/https://www.worldscientific.com/doi/pdf/10.1142/9789813272880_0188}
  {\path{arXiv:https://www.worldscientific.com/doi/pdf/10.1142/9789813272880_0188}},
  \href {https://doi.org/10.1142/9789813272880_0188}
  {\path{doi:10.1142/9789813272880_0188}}.

\bibitem{Williams2020}
Virginia~Vassilevska Williams and Yinzhan Xu.
\newblock Monochromatic triangles, triangle listing and {APSP}.
\newblock In Sandy Irani, editor, {\em 61st {IEEE} Annual Symposium on
  Foundations of Computer Science, {FOCS} 2020, Durham, NC, USA, November
  16-19, 2020}, pages 786--797. {IEEE}, 2020.
\newblock \href {https://doi.org/10.1109/FOCS46700.2020.00078}
  {\path{doi:10.1109/FOCS46700.2020.00078}}.

\bibitem{Yildiz2011-prob}
Hakan Yıldız, Luca Foschini, John Hershberger, and Subhash Suri.
\newblock The union of probabilistic boxes: Maintaining the volume.
\newblock In Camil Demetrescu and Magn{\'{u}}s~M. Halld{\'{o}}rsson, editors,
  {\em Algorithms - {ESA} 2011 - 19th Annual European Symposium,
  Saarbr{\"{u}}cken, Germany, September 5-9, 2011. Proceedings}, volume 6942 of
  {\em Lecture Notes in Computer Science}, pages 591--602. Springer, 2011.
\newblock \href {https://doi.org/10.1007/978-3-642-23719-5\_50}
  {\path{doi:10.1007/978-3-642-23719-5\_50}}.

\bibitem{Yildiz2011}
Hakan Yıldız, John Hershberger, and Subhash Suri.
\newblock A discrete and dynamic version of {K}lee's measure problem.
\newblock In {\em Proceedings of the 23rd Annual Canadian Conference on
  Computational Geometry, Toronto, Ontario, Canada, August 10-12, 2011}, 2011.
\newblock URL: \url{http://www.cccg.ca/proceedings/2011/papers/paper28.pdf}.

\end{thebibliography}

\end{document}